\documentclass[11pt,a4paper]{article}
\usepackage[utf8]{inputenc}
\usepackage[margin=1.2in]{geometry}
\usepackage{enumitem}
\usepackage{algorithm}
\usepackage[noend]{algpseudocode}
\usepackage{nicefrac}
\usepackage{graphicx}
\usepackage[fleqn]{amsmath}
\usepackage{amsfonts,amssymb,amsthm}
\usepackage{enumitem}
\usepackage{dsfont}
\usepackage{hyperref}
\usepackage{microtype}

\newtheorem{claim}{Claim}
\newcommand{\poly}{\operatorname{poly}}

\newcommand{\R}{\mathbb{R}}
\newcommand{\N}{\mathbb{N}}
\newcommand{\scX}{\mathcal{X}}
\newcommand{\nil}{\textsc{nil}}
\newcommand{\Ind}{\mathds{1}}
\newcommand{\scD}{\mathcal{D}}
\newcommand{\scY}{\mathcal{Y}}

\newcommand{\tvd}[2]{\operatorname{tvd}(#1,#2)}  
\newcommand{\ED}{\triangle}
\newcommand{\EDR}{\ED^{\!*}}

\newcommand{\beps}{\pmb{\epsilon}}
\newcommand{\epsThr}{\alpha}
\newcommand{\epsApx}{\beta}

\newcommand{\Vr}{R} 

\newcommand{\AlgoRebuild}{\textsc{delay-apx}}
\newcommand{\AlgoMain}{\textsc{fudy-wc}}
\newcommand{\Instance}{\textsc{inst}}
\newcommand{\greedy}{\textsc{greedy}}
\newcommand{\splfun}{\ensuremath{\rho}} 
\newcommand{\spl}{\ensuremath{\sigma}} 
\newcommand{\lab}{\lambda} 
\newcommand{\cost}{f} 
\newcommand{\gain}{\ensuremath{G}}
\newcommand{\giniImp}{\ensuremath{I_{\operatorname{gini}}}} 
\newcommand{\giniGain}{\ensuremath{G_{\operatorname{gini}}}} 
\newcommand{\IG}{\ensuremath{G_{\operatorname{info}}}} 
\newcommand{\VarGain}{\ensuremath{G_{\Var}}} 
\newcommand{\Var}{\operatorname{var}}
\newcommand{\Lab}{\mathcal{L}}
\newcommand{\Spl}{\mathcal{S}}

\renewcommand{\lg}{\log}

\newtheorem{theorem}{Theorem}
\newtheorem{definition}{Definition}
\newtheorem{lemma}{Lemma}

\begin{document}
\title{Fully-Dynamic Approximate Decision Trees With Worst-Case Update Time Guarantees}
\author{
   Marco Bressan\\ Department of Computer Science,\\ University of Milan
   \and
   Mauro Sozio\\ Institut Polytechnique de Paris,\\ Télécom Paris
}

\maketitle

\begin{abstract}
    We give the first algorithm that maintains an approximate decision tree over an arbitrary sequence of insertions and deletions of examples, with strong guarantees on the \emph{worst-case} running time per update request. For instance, we show how to maintain a decision tree where every vertex has Gini gain within an additive $\alpha$ of the optimum by performing $O\!\left(\frac{d\,\log^4 n}{\alpha^3}\right)$ elementary operations per update, where $d$ is the number of features and $n$ the maximum size of the active set (the net result of the update requests). We give similar bounds for the information gain and the variance gain. In fact, all these bounds are corollaries of a more general result stated in terms of \emph{decision rules} --- functions that, given a set $S$ of examples, decide whether to split $S$ or predict a label. Decision rules give a unified view of greedy decision tree algorithms regardless of the example and label domains, and lead to a general notion of \emph{$\epsilon$-approximate decision trees} that, for natural decision rules such as those used by ID3 or C4.5, implies the gain approximation guarantees above. The heart of our work provides a deterministic algorithm that, given any decision rule and any $\epsilon > 0$, maintains an $\epsilon$-approximate tree using $O\!\left(\frac{d\, \cost(n)}{n} \poly\frac{h}{\epsilon}\right)$ operations per update, where $\cost(n)$ is the complexity of evaluating the rule over a set of $n$ examples and $h$ is the maximum height of the maintained tree.
\end{abstract}

\section{Introduction}
Decision trees represent a fundamental class of models in machine learning, and are among the most successful tools in data mining~\cite{shalevshwartz2014understanding,top10algo}. Given a feature set $\scX$ (e.g., $\R^d$) and a label set $\scY$ (e.g., $\{1,\ldots,k\}$), a decision tree over $\scX, \scY$ is a rooted binary tree $T$ where each vertex $v$ satisfies the following constraints. If $v$ is internal then it has two children, and it has an associated \emph{split rule} $\spl_v : \scX \to \{0,1\}$ that, for each each $x \in \scX$, specifies to which child $x$ should be sent. If $v$ is a leaf, then it has an associated \emph{labeling rule} $\lab_v : \scX \to \scY$ that for every $x \in \scX$ predicts some $y \in \scY$. For every $x \in \scX$, the prediction $T(x)$ of the tree is the label of the leaf reached by $x$ by following the split rules starting from the root of $T$. The problem of constructing a ``good'' decision tree asks, given a (multi)set $S$ of labeled examples from $\scX \times \scY$, to compute a decision tree $T$ that optimizes some measure of quality.\footnote{Formally, in this work we always consider \emph{multisets} of examples, since we do not care about order. However, for easiness  we may refer to them as to ``sets'' or ``sequences'', for instance by writing $S \in (\scX \times \scY)^*$.} For many such measures the problem turns out to be hard~\cite{HYAFIL197615,shalevshwartz2014understanding}, but practice has shown that greedy construction heuristics such as ID3~\cite{Quinlan86-inductive} or C4.5~\cite{quinlan1993c4} yield decision trees that are good enough; in fact, ``decision tree'' has almost become a synonym for \emph{greedy} decision tree.
The typical greedy heuristic proceeds as follows. Given $S$, one computes a split rule $\spl$ that partitions $S$ into two subsets $(S_0,S_1)$ so to maximize some measure of gain, such as the Gini or information gain (for classification) or the variance gain (for regression). If the resulting gain is below some threshold, say $\alpha=0.1$, then the tree consists of a single vertex $v$ with an associated labeling rule $\lab_v$ --- such as the majority or average label of $S$. If instead the gain is at least $\alpha$, then the tree has a root $r$ with split rule $\spl_r=\spl$ whose subtrees are obtained by recursing on $S_0$ and $S_1$. One can add further stopping conditions, such as the tree reaching a certain height or the number of examples at vertices falling below some threshold.

Owing to the increasing rate at which data is generated, updated, and deleted, there has been growing interest in \emph{fully-dynamic} algorithms. In this setting, the input is a sequence of \emph{update requests} in the form $(s,o)$ where $s$ is a labeled example and $o \in \{\textsc{ins},\textsc{del}\}$ is a request to insert or delete $s$ in the current \emph{active set}  --- the multiset of examples obtained by executing all requests received so far. The goal of the algorithm is to maintain a learning model that is good w.r.t.\ the current active set while performing as few elementary operations as possible between any two consecutive update requests. Dynamic algorithms have been studied for maintaining solutions to facility location and clustering~\cite{DBLP:conf/nips/Cohen-AddadHPSS19,DBLP:conf/esa/HenzingerK20,DBLP:journals/corr/abs-2112-07050, DBLP:conf/www/ChanGS18}, as well as for maintaining approximations of the densest subgraph~\cite{DBLP:conf/stoc/BhattacharyaHNT15,DBLP:conf/www/EpastoLS15, DBLP:conf/stoc/SawlaniW20} and accurate subgraph counts~\cite{DBLP:journals/tkdd/StefaniERU17}.

Perhaps surprisingly, there have been a relatively small number of studies on fully-dynamic algorithms for supervised machine learning problems, with fully-dynamic decision trees being somehow neglected until~\cite{BDS22}. In that work, the authors consider the case $\scY=\{0,1\}$ and aim at maintaining a decision tree that is close to the one produced by greedy algorithms such as ID3 or C4.5, where closeness is measured by what they call \emph{$\beps$-feasibility}. Let us introduce some further notation. Given a decision tree $T$, a set $S$ of labeled examples, and a vertex $v \in V(T)$, let $S(T,v)$ be the subset of $S$ formed by those examples that reach $v$ when following the split rules of $T$. A \emph{gain} is a function $G$ that maps pairs $(S,\spl)$ to non-negative reals. For any $v \in V(T)$ let $\spl^*_v = \arg \max_{\spl} G(S(T,v),\spl)$ be the split rule with maximum gain over $S(T,v)$. Given $\beps=(\alpha,\beta) \in (0,1]^2$, we say $T$ is $\beps$-feasible w.r.t.\ $S$ if every $v \in V(T)$ satisfies what follows: if $G(S(T,v),\spl_v^*) \ge \alpha$ then $v$ is internal and $G(S(T,v),\spl_v) \ge G(S(T,v),\spl_v^*) - \beta$, and if $G(S(T,v),\spl_v^*) =0$ then $v$ is a leaf and $\lab_v$ is a majority label.\footnote{In this work by ``majority label'' we mean the mode, i.e., a label with highest frequency.} This means that the algorithm must find a near-optimal split if its gain is sufficiently large, while it must create a leaf if all splits have zero gain; in all other cases, the algorithm has carte blanche. A dynamic algorithm is $\beps$-feasible if at every time it maintains a decision tree that is $\beps$-feasible w.r.t.\ the current active set. Although defined formally in~\cite{BDS22}, $\beps$-feasibility was targeted implicitly by incremental heuristic decision tree algorithms such as Hoeffding Trees~\cite{Domingos00-HighSpeedStreams}.

The main result of~\cite{BDS22} is a deterministic algorithm that maintains an $\beps$-feasible tree over arbitrary sequences of insertions and deletions using an amortized number of operations per request in $O\!\left(\frac{d \log^3 n}{\epsilon^2}\right)$, where $n$ is the maximum size of the active set at any time and $\epsilon=\min(\alpha,\beta)$. We shall describe briefly their algorithm. The key observation is that, if a set of examples $S$ incurs at most $\epsilon|S|$ insertions and deletions, then the Gini gain of any split rule $\spl$ over $S$ changes by $O(\epsilon)$. Hence, a vertex $v$ can violate the $\beps$-feasibility only if it has been reached by $\Omega(\epsilon n_v)$ update requests from the last time its subtree $T_v$ was rebuilt, with $n_v= |S(T,v)|$. Therefore, if rebuilding $T_v$ requires $O(d \,h \,n_v \,\log n_v)$ operations where $h$ is the height of the output tree (which is the case for algorithms like ID3 or C4.5), then each one of those $\Omega(\epsilon n_v)$ update requests pays for $\Omega\!\left(\frac{d\, h \,\log n_v}{\epsilon}\right)$ operations. Observing that in an $\beps$-feasible tree $h = O\!\left(\frac{\log n}{ \epsilon}\right)$ leads, with some further arguments, to the bound mentioned above.
Unfortunately, this technique gives only an \emph{amortized} bound.
This is not an artifact of the analysis: if the current active set has size $n$, within the next $\epsilon n$ update requests the algorithm will recompute the entire tree --- and it is not clear how to ``spread'' this computation over those requests. Therefore, it is not obvious that there exists a fully-dynamic algorithm that maintains an $\beps$-feasible tree while using only $O( d \poly \log n)$ operations per update request.
In this work we give the first such algorithm; in fact, we give an algorithm that yields significantly more general guarantees.


\subsection{Contributions}

\paragraph{(1) Approximation notions for greedy decision trees.} We introduce a general notion of approximation for decision trees. We start by defining \emph{decision rules}, functions $\splfun$ that send every set $S$ of labeled examples into either a split rule $\spl$ or a labeling rule $\lab$. Any decision rule $\splfun$ defines a greedy algorithm $\greedy_{\splfun}$ that, on input $S$, applies $\splfun$ to construct the root of the decision tree and then proceeds recursively. Popular algorithms such as ID.3 or C4.5 are $\greedy_{\splfun}$ for particular choices of $\splfun$. We often take as an example the decision rule $\splfun$ that considers all split rules in the form $\spl(x)=\Ind_{x_j < t}$, that is, ``is the $j$-th feature of $x$ smaller than $t$?''; if the best such split rule yields Gini gain at least $\alpha$ for some fixed $\alpha > 0$, then $\splfun$ returns that rule, else it returns a majority label. 
Recall $S(T,v)$ from above. Given two (multi)sets $S,S'$, their edit distance or symmetric difference is $\ED(S,S')=|S \setminus S'|+|S'\setminus S|$, and their \emph{relative edit distance} is $\EDR(S,S')=\frac{\ED(S,S')}{\max(|S|,|S'|)}$. We say $T$ is $\epsilon$-approximate w.r.t.\ $\greedy_{\splfun}$ on $S$, or w.r.t.\ $(S,\splfun)$ for short, if for every $v \in V(T)$ there exists $S_v$ such that $\EDR(S(T,v),S_v) \le \epsilon$ and that $\splfun(S_v)$ is precisely $\spl_v$ if $v$ is internal and $\lab_v$ if $v$ is a leaf. That is, the decision to make $v$ internal or leaf, and the corresponding rule $\spl_v$ or $\lab_v$, are the output of $\splfun$ on a set that is close to $S(T,v)$.
This notion of approximation is independent of the gain and we can show that, as a special case, it yields the $\beps$-feasibility of~\cite{BDS22}.

\paragraph{(2) A fully-dynamic algorithm for approximate decision trees.}
Equipped with the notions of decision rule $\splfun$ and of $\epsilon$-approximation, we consider what we call the \emph{dynamic $\epsilon$-approximate decision tree problem}. Let $U$ be a sequence of update requests. The problem asks to maintain for all $i=1,2,\ldots$ a decision tree $T^i$ that is $\epsilon$-approximate w.r.t.\ $(S^i,\splfun)$ where $S^i$ is the active set defined by the first $i$ requests in $U$. We present a deterministic algorithm for this problem, \AlgoMain, that uses 
$O\!\left( h_{\splfun}(\epsilon,n)^2 \, \log n \cdot \left(h_{\splfun}(\epsilon,n) + \frac{d \log n}{\epsilon} + \frac{\cost_{\splfun}(n)}{\epsilon n}\right) \!\right)$
operations per update request, where $h_{\splfun}(\epsilon,n)$ is the maximum height of any $\epsilon$-approximate tree w.r.t.\ $(S,\splfun)$ if $|S|\le n$, $\cost_{\splfun}(n)$ is the cost of computing $\splfun$ on a set of $n$ examples, and $n=\max_{i \ge 1}|S^i|$. To appreciate this bound, suppose $\cost_{\splfun}(n)=O(d n \log n)$ and $\splfun$ produces $\gamma$-balanced splits (ones where each subset contains a fraction at least $\gamma$ of the examples). Then we can prove that the bound above is in $O\!\left(\frac{d\, \log^4 n}{\epsilon^3}\right)$.
Our algorithm \AlgoMain\ is substantially different from the algorithm of~\cite{BDS22}. It is based on solving, simultaneously at each vertex of the tree, a relaxed version of the problem, which we call \emph{delayed approximate decision tree problem}. The input to this relaxed problem consists of a set of examples $S$ and a sequence $U$ of $\epsilon |S|$ update requests, and the goal is to compute a decision tree that is $\epsilon$-approximate w.r.t.\ $(S',\splfun)$, where $S'$ is the active set defined by $S$ and $U$, using as few operations as possible for each update request. We construct an algorithm that solves this problem while using $O\!\left( h \, \log n \cdot \left(h + \frac{d \log n}{\epsilon} + \frac{\cost_{\splfun}(n)}{\epsilon n}\right) \right)$ operations per request, where $h \le h_{\splfun}(\epsilon,n)$. From a technical point of view, this is the heart of our work.

\paragraph{(3) Max-gain approximation of popular trees in polylog update time.}
As an application of our results, we give fully-dynamic algorithms with worst-case update time $O\!\left(d \cdot \poly\frac{\log(n)}{\epsilon}\right)$ for maintaining a decision tree where every vertex has gain within an additive $\epsilon$ of the optimum, in the sense of $\beps$-feasibility. Let $\scD$ be the set of all split rules in the form $\Ind_{x_j < t}$ or $\Ind_{x_j = t}$, and let $G$ be a gain function.
We show that, if $G$ is the Gini gain, then we can maintain an $\beps$-feasible tree using $O\!\left(\frac{d\,\log^4 n}{\epsilon^3}\right)$ operations per update request, where $\epsilon=\min(\alpha,\beta)$. This bound is only a factor $\frac{\log n}{\epsilon}$ higher than the one of~\cite{BDS22}, and is for the worst case rather than just amortized. For the information gain and the variance gain we obtain bounds of $O\!\left(\frac{d\,\log^7 n}{\epsilon^3}\right)$ and $O\!\left(\frac{d\,c^2\,\log^4 n}{\epsilon^3}\right)$ respectively, where for the variance gain we assume $\scY = [-c,c]$ for $c\in\R$. To prove these bounds we analyze \emph{max-gain threshold decision rules} $\splfun$ --- ones that, whenever $\splfun(S) \in \scD$, then $\splfun(S)=\arg\max_{\spl \in \scD}G(S,\spl)$, and that $\splfun(S) \in \scD$ if and only if $\max_{\spl \in \scD}G(S,\spl) \ge \alpha$ for some fixed $\alpha > 0$. We show that, for a certain $\epsilon$ function of $\beps$, a tree that is $\epsilon$-approximate w.r.t.\ a max-gain threshold rule is also $\beps$-feasible. Coupling this with our guarantees for $\epsilon$-approximate trees yields the bounds above.

\subsection{Related work}
Except~\cite{BDS22}, the only existing dynamic algorithms for decision trees are incremental: they receive a stream of labeled examples and maintain a decision tree that performs well compared to the tree built on the examples seen so far. The first such algorithms were Hoeffding trees~\cite{Domingos00-HighSpeedStreams}, which spurred a line of research on trees that adapt to so-called concept drifts~\cite{Domingos01-MiningTimeSeries,Gama03-HighSpeedStreams,Manapragada18,Das19,ijcai2020-177,Haug22,Jin03-Streaming,Rutkowski13-MiningStreams}; see~\cite{Manapragada2020} for a survey. Unfortunately, all those algorithms assume the examples are i.i.d., which allows them to compute splits that have nearly-maximum gain with high confidence. Moreover, none of those algorithms supports deletion, and they have a worst-case update time as large as $\Omega(n)$ when $\scY=\R$.
\\[10pt]

\subsection{Organisation of the manuscript}
Section~\ref{sec:prelim} pins down definitions and notation. Section~\ref{sec:dyn} introduces the dynamic approximate decision tree problem and reduces it to the delayed approximate decision tree problem, which is then studied in Section~\ref{sec:delayed}; these two sections contain all our main results. Section~\ref{sec:smooth} proves several properties of common gain functions and related decision rules needed by our main claims. All missing parts can be found in the Appendix.

\section{Preliminaries}\label{sec:prelim}
We assume $\scX=\scX_1 \times \ldots \times \scX_d$ where each $\scX_j$ is a totally ordered set such that $t<t'$ can be evaluated in time $O(1)$ for every $t,t' \in \scX_j$; this defines a total order on $\scX$ where $x<x'$ can be evaluated in time $O(d)$ for all $x,x' \in \scX$.
Any $x \in \scX$ is an \emph{unlabeled example}, and any $s=(x,y) \in \scX \times \scY$ is a \emph{labeled example}. Any $S=(s_1,\ldots,s_n) \in (\scX \times \scY)^*$ is a multiset, or simply set, of labeled examples; we always assume $|S| \ge 1$, which does not affect our results since obviously on any $S$ of bounded size one can compute any desired decision tree in time $O(1)$.
An \emph{update request} is a pair $(s,o)$ where $o \in \{\textsc{ins},\textsc{del}\}$. For any $S\in(\scX \times \scY)^*$, possibly empty, and any sequence of update requests $U$, the \emph{active set} determined by $S$ and $U$, denoted $S+U$, is the set of labeled examples obtained by applying $U$ to $S$ in the obvious way. For $s \in \scX \times \scY$ define $S+s=S+((s,\textsc{ins}))$. 
\\[5pt]
\noindent\textbf{Split rules, labeling rules, and decision trees.}
We assume that for every $x \in \scX$ and every split rule $\spl$ or labeling rule $\lab$ one can compute $\spl(x)$ and $\lab(x)$ in time $O(1)$. For every $S\in(\scX \times \scY)^*$ we let $\spl(S)=(S_0,S_1)$ where $S_0=\{(x,y) \in S : \spl(x)=0\}$ and $S_1=\{(x,y) \in S : \spl(x)=1\}$. Clearly one can compute $\spl(S)$ in time $O(|S|)$.
Let $T$ be a decision tree over $\scX,\scY$. For $x \in\scX$ we denote by $P(T,x)$ the path whose first vertex is the root of $T$ and where each internal vertex $v$ is followed by its own left child if $\spl_v(x)=0$ and its own right child if $\spl_v(x)=1$. We say $x$ \emph{reaches} $v \in V(T)$ if $v \in P(T,x)$. We let $T(x)=\lab_v(x)$ where $v$ is the only leaf in $P(T,x)$. Clearly one can compute $P(T,x)$ and $T(x)$ in time $O(h(T))$ where $h(T)$ is the height of $T$. All these definitions extend naturally to labeled examples and update requests.
All decision trees $T$ in this work are pointer-based: every $v \in V(T)$ is represented by a data structure holding pointers to $v$'s parent and children (if any). We denote by $D$ a generic associative array data structure that supports insertion, lookup, and deletion in time $O(\log |D|)$, where $|D|$ is the total number of entries in $D$, as well as enumeration in time $O(|D|)$; this can be fulfilled by a self-balancing search tree. Each vertex $v$ of $T$ points to such an array $D(T,v)$ that stores some set of labeled examples by mapping each example to the number of its occurrences. We may keep additional counters or structures at $v$, such as the size of the set in $D(T,v)$.
\\[5pt]
\noindent\textbf{Decision rules and greedy algorithms.}
Let $\Spl$ be a family of split rules and $\Lab$ a family of labeling rules. A \emph{decision rule} is a map:
\begin{align}
    \splfun : (\scX \times \scY)^* \to \Spl \cup \Lab
\end{align}
We say $\splfun$ is $\gamma$-balanced if $\min(|S_0|,|S_1|) \ge \gamma |S|$ whenever $\splfun(S)\in\Spl$.
We denote by $\cost_{\splfun} : \R_{\ge 0} \to \R_{\ge 0}$ an upper bound on the cost of computing $\splfun(S)$ as a function of $|S|$. We assume $\cost_{\splfun}$ is twice differentiable and $\cost_{\splfun}', \cost_{\splfun}'' \ge 0$; this implies $\cost_{\splfun}$ is superadditive and $\cost_{\splfun}(n)=\Omega(n)$, so in time $\cost_{\splfun}(|S|)$ one can compute both $\splfun(S)$ and $\spl(S)$. The greedy algorithm $\greedy_{\splfun}$ computes a decision tree $T=\greedy_{\splfun}(S)$ as follows. If $\splfun(S) \in \Lab$, then the root $r$ of $T$ is a leaf and $\lab_r = \splfun(S)$. Otherwise $r$ is internal, $\spl_r = \splfun(S)$, and the subtrees of $r$ are $\greedy_{\splfun}(S_0)$ and $\greedy_{\splfun}(S_1)$ where $(S_0,S_1)=\spl_r(S)$ and $\max(|S_0|,|S_1|) \le |S|-1$. Moreover, $S(T,v)$ is stored in $D(T,v)$. Note that $|V(T)|=O(|S|)$. 
\\[5pt]
\noindent\textbf{Gain functions, max-gain rules, and threshold rules.}
Let $\Spl$ be a family of split rules. A \emph{gain} is a function $\gain : (\scX \times \scY)^* \times \Spl \to \R_{\ge 0}$; intuitively $\gain(S,\spl)$ measures the quality of $\spl(S)$. We say a decision rule $\splfun : (\scX \times \scY)^* \to \Spl \cup \Lab$ is a max-$\gain$ rule if:
\begin{align}
    \splfun(S) \in \Spl \implies \splfun(S) \in \arg \max_{\spl \in \Spl} \gain(S,\spl) \quad \forall\,S \in (\scX \times \scY)^*
\end{align}
We say $\splfun$ has \emph{threshold} $\alpha > 0$ if:
\begin{align}
    \splfun(S) \in \Spl \; \iff \; \max_{\spl \in \Spl} \gain(S,\spl) \ge \alpha \quad \forall\,S \in (\scX \times \scY)^*
\end{align}
Let $g : (\scX \times \scY)^* \to \R_{\ge 0}$. A gain $G$ is a \emph{conditional} $g$-gain if:
\begin{align}
    G(S,\spl) = g(S) - \left( \frac{|S_0|}{|S|} g(S_0) +  \frac{|S_1|}{|S|} g(S_1) \right) \quad \forall\,S \in (\scX \times \scY)^*, \spl \in \Spl
\end{align}
where $(S_0,S_1)=\spl(S)$. In this work we consider the Gini gain $\giniGain$, the information gain $\IG$, and the variance gain $\VarGain$, which are the conditional $g$-gains for $g$ being respectively the Gini impurity, the entropy, and the variance of the labels --- see Appendix~\ref{apx:gains}.

\section{Dynamic approximate decision trees}\label{sec:dyn}
A \emph{dynamic decision tree algorithm} receives a sequence of update requests $\{(s_i,o_i)\}_{i \ge 1}$ and satisfies the following constraint: for all $i \ge 1$, after $(s_i,o_i)$ arrives and before $(s_{i+1},o_{i+1})$ arrives, there exists a decision tree $T^i$ such that for every $x \in \scX$ the algorithm can compute $T^i(x)$ in time $O(h(T^i))$.\footnote{Our algorithms satisfy the $O(h(T))$ constraint naturally, but if needed one can always relax it by allowing for, say, additional $d$ or $\poly\log(i)$ factors.} We also say that the algorithm \emph{maintains} $T^i$ at time $i$. This section describes a dynamic decision tree algorithm that maintains a tree close to the one that $\greedy_{\splfun}$ would produce on the current active set. To start with, we shall formalize what one means by ``close''.
\begin{definition}\label{def:approx_tree}
Let $S \in (\scX \times \scY)^*$ and let $\splfun$ be a decision rule. A decision tree $T$ is $\epsilon$-approximate w.r.t.\ $(S,\splfun)$ if for every $v \in V(T)$ there exists $S_v\in (\scX \times \scY)^*$ such that $\EDR(S(T,v),S_v) \le \epsilon$ and $\splfun(S(T,v)) = \splfun(S_v)$.
\end{definition}
\noindent The following is the central problem of this work.
\begin{definition}\label{def:main_problem}
The \emph{dynamic approximate decision tree problem} asks, given in input a decision rule $\splfun$, a real $\epsilon\in(0,1)$, and a sequence of update requests $\{(s_i,o_i)\}_{i \ge 1}$, to maintain for all $i \ge 1$ a tree $T^i$ that is $\epsilon$-approximate w.r.t.\ $(S^i,\splfun)$ where $S^i$ is the active set determined by $(s^1,o^1),\ldots,(s^i,o^i)$.
\end{definition}
\noindent Let $A$ be any algorithm for the dynamic approximate decision tree problem. The number of \emph{operations per update request} performed by $A$ is the maximum number of elementary operations $A$ performs between any two consecutive update requests, as well as before the first one and after the last one. Let $h_{\splfun}(\epsilon,n)$ be the maximum height of any decision tree that is $\epsilon$-approximate w.r.t.\ $(S,\splfun)$ for $|S| \le n$; one can prove that $h_{\splfun}(\epsilon,n)=O\!\left(\frac{\log n}{\epsilon}\right)$ for several natural $\splfun$. Our main result is:
\begin{theorem}\label{thm:main_h}
There is an algorithm \AlgoMain\ for the dynamic approximate decision tree problem that uses a number of operations per update request in
$$O\!\left( h_{\splfun}(\epsilon,n)^2 \, \log n \cdot \left(h_{\splfun}(\epsilon,n) + \frac{d \log n}{\epsilon} + \frac{\cost_{\splfun}(n)}{\epsilon n}\right) \right)$$ where $n=\max_{i \ge 1}|S^i|$.
\end{theorem}
\noindent Note that $h_{\splfun}(\epsilon,n)$ is the worst-case cost of \emph{predicting} the label of an example in an $\epsilon$-approximate tree built on $n$ examples. Therefore, if for instance $\cost_{\splfun}(n)=O(d\, n \log n)$ --- which holds for decision rules like those used by ID3 or C4.5, see Section~\ref{sec:complexity} --- then our algorithm has a cost per update that is at most $O\!\left(h_{\splfun}(\epsilon,n) \log n \cdot \left(h_{\splfun}(\epsilon,n) + \frac{d \log n}{\epsilon} \right)\! \right)$ times the worst-case cost of predicting a label.

By combining Theorem~\ref{thm:main_h} with the fact that $h_{\splfun}(\epsilon,n)=O\!\left(\frac{\log n}{\gamma-2\epsilon}\right)$ for all $\gamma$-balanced decision rules $\splfun$ (Lemma~\ref{lem:balanced_approx}) we obtain:
\begin{theorem}\label{thm:main}
There is an algorithm \AlgoMain\ for the dynamic approximate decision tree problem that, if $\splfun$ is $\gamma$-balanced and $\epsilon \in \big(0, \frac{\gamma}{3}\big]$, uses a number of operations per update request in
$$O\!\left( \frac{\log^3 n}{\epsilon^3} \cdot \left(d \log n + \frac{\cost_{\splfun}(n)}{n}\right) \right)$$
where $n=\max_{i \ge 1}|S^i|$.
\end{theorem}
\noindent If again $\cost_{\splfun}(n)=O(d\, n \log n)$, then the bound of Theorem~\ref{thm:main} is in $O\!\left(\frac{d \log^4 n}{\epsilon^3}\right)$. For popular decision rules this provides the first fully-dynamic algorithm that maintains a tree with near-optimal gain at every vertex using $O\!\left(d \poly\frac{ \log n}{\epsilon}\right)$ operations per update request, see Section~\ref{sec:appl}.

\subsection{Applications}\label{sec:appl}
By leveraging Theorem~\ref{thm:main} we generalise the guarantees of~\cite{BDS22} while strengthening them from amortized to worst-case. To this end, we first rewrite the definition of $\beps$-feasibility, generalising it to arbitrary domains, arbitrary gains, and to both classification and regression trees. For the moment we ignore their ``pruning thresholds'' that force vertices at a prescribed depth or with few enough examples to be leaves; we show below that such constraints can be satisfied at no additional cost.
\begin{definition}\label{def:eps_feasibility}
Let $\beps=(\epsThr,\epsApx)\in (0,1]^2$, let $\Spl$ be a set of split rules, let $G$ be a gain function, and let $S\in(\scX\times\scY)^*$. A decision tree $T$ is $\beps$-feasible w.r.t.\ $(\Spl,G,S)$ if for every $v \in V(T)$:
\begin{enumerate}\itemsep0pt
\item if $G(S(T,v)) = 0$ then $v$ is a leaf, and if $G(S(T,v)) \ge \epsThr$ then $v$ is an internal node
\item if $v$ is internal then $G(S(T,v),\spl_v) \ge \max_{\spl \in \Spl}G(S(T,v),\spl) - \epsApx$
\item if $v$ is a leaf then $\lab_v$ is a majority label or the average of the labels of $S(T,v)$.
\end{enumerate}
\end{definition}
\noindent An algorithm maintains an $\beps$-feasible decision tree if for all $i \ge 1$ it maintains a tree $T^i$ that is  $\beps$-feasible w.r.t.\ $(\Spl,G,S^i)$.

Next, we show that $\epsilon$-approximation subsumes $\beps$-feasibility, in the following sense.
\begin{theorem}\label{thm:eps_to_beps}
    Let $\Spl$ be a set of split rules, let $G\in\{\giniGain,\IG,\VarGain\}$, and let $\beps=(\alpha,\beta) \in (0,1]^2$. Let $\splfun : (\scX\times\scY)^* \to \scD \cup \Lab$ be a max-$G$ decision rule with threshold $\frac{\alpha}{2}$ that assigns majority/average labels, and define:
    \begin{align*}
        \epsilon = \left\{
        \begin{array}{ll}
            \frac{\min(\alpha,\beta)}{100} & G=\giniGain \\[5pt]
            \frac{\min(\alpha,\beta)}{130 \log n} ~~& G=\IG \\[5pt]
            \frac{\min(\alpha,\beta)}{80c^2} & G=\VarGain, \; c = \sup_{y \in \scY} |y|/2
        \end{array}
        \right.
    \end{align*}
    Then, for all $S \in (\scX\times\scY)^*$, every decision tree $T$ that is $\epsilon$-approximate w.r.t.\ $(S,\splfun)$ is $\beps$-feasible w.r.t.\ $(\Spl,G,S)$.
\end{theorem}
\begin{proof}
    Suppose first $G=\giniGain$. Let $v \in V(T)$. By definition of $\epsilon$-approximation (Definition~\ref{def:approx_tree}) there exists $S_v$ such that $\EDR(S(T,v),S_v) \le \epsilon$ and that the decision taken at $v$ is $\splfun(S_v) \in \Spl$. Let $S_v$ be any such set. Let $\spl_v$ and $\spl_v^*$ be split rules with maximum gain on respectively $S_v$ and $S(T,v)$:
    \begin{align}
        \spl_v &= \arg\max_{\spl \in \Spl}G(S_v,\spl)
        \\
        \spl_v^* &= \arg\max_{\spl \in \Spl}G(S(T,v),\spl)
    \end{align}
    Suppose $G(S(T,v),\spl_v^*) \ge \alpha$. By Theorem~\ref{thm:smooth_gains}, and by the choice of $\epsilon$ and the definition of $\spl_v$ and $\spl_v^*$:
    \begin{align}
        G(S_v,\spl_v) \ge G(S_v,\spl_v^*) \ge G(S(T,v),\spl_v^*) - 48 \epsilon > G(S(T,v),\spl_v^*) - \frac{\alpha}{2} \ge \frac{\alpha}{2}
    \end{align}
    Similarly, if $G(S(T,v),\spl_v^*) = 0$ then:
    \begin{align}
        G(S_v,\spl_v) \le G(S(T,v),\spl_v) + 48 \epsilon < G(S(T,v),\spl_v) + \frac{\alpha}{2} \le G(S(T,v),\spl_v) + \frac{\alpha}{2}
    \end{align}
    and the rightmost expression equals $\frac{\alpha}{2}$. Since $\splfun$ has threshold $\frac{\alpha}{2}$, this proves that $v$ is internal if $G(S(T,v),\spl_v^*) \ge \alpha$ and $v$ is a leaf if $G(S(T,v),\spl_v^*) = 0$. Finally, if $v$ is internal then by Lemma~\ref{lem:gain_approx}:
    \begin{align}
        G(S(T,v),\spl_v) \ge G(S(T,v),\spl_v^*) - 96\epsilon > G(S(T,v),\spl_v^*) - \beta    
    \end{align}
    Since the facts above hold for any choice of $S_v$, we conclude that $T$ is $\beps$-feasible as desired.

    The proof for $G=\IG$ is similar. If $G(S(T,v),\spl_v^*) \ge \alpha$ then again by Theorem~\ref{thm:smooth_gains} and the definition of $\spl_v$ and $\spl_v^*$:
    \begin{align}
        G(S_v,\spl_v) \ge G(S_v,\spl_v^*) \ge G(S(T,v),\spl_v^*) - 60 \epsilon \log n > G(S(T,v),\spl_v^*) - \frac{\alpha}{2} \ge \frac{\alpha}{2}
    \end{align}
    Similarly, if $G(S(T,v),\spl_v^*) = 0$ then:
    \begin{align}
        G(S_v,\spl_v) \le G(S(T,v),\spl_v) + 60 \epsilon \log n < G(S(T,v),\spl_v) + \frac{\alpha}{2} \le G(S(T,v),\spl_v) + \frac{\alpha}{2}
    \end{align}
    Finally, if $v$ is internal then by Lemma~\ref{lem:gain_approx}:
    \begin{align}
        G(S(T,v),\spl_v) \ge G(S(T,v),\spl_v^*) - \epsilon \cdot 120 \log n > G(S(T,v),\spl_v^*) - \beta    
    \end{align}

    The proof for $G=\VarGain$ is completely analogous.
\end{proof}

In the special case where $\scY=\{0,1\}$, $\Spl$ is the set of split rules in the form $\Ind_{x_j < t}$, $G=\giniGain$, and $\lab_v$ is a majority label,~\cite{BDS22} maintain an $\beps$-feasible decision tree using an \emph{amortized} $O\!\left(\frac{d \log^3 n}{\min(\alpha,\beta)^2}\right)$ operations per update. We prove:
\begin{theorem}\label{thm:fudy_eps_feasible}
Let $\beps=(\alpha,\beta) \in (0,1]^2$, let $\Spl$ be the set of split rules in the form $\Ind_{x_j < t}$ or $\Ind_{x_j = t}$, and let $G\in\{\giniGain,\IG,\VarGain\}$. There is an algorithm that maintains an $\beps$-feasible decision tree using a worst-case number of operations per update request in:
\begin{align*}
    \begin{array}{ll}
        O\!\left( \frac{d \log^4 n}{\min(\alpha,\beta)^3} \right) & \text{if }G=\giniGain \\[6pt]
        O\!\left( \frac{d \log^7 n}{\min(\alpha,\beta)^3} \right) & \text{if }G=\IG \\[6pt]
        O\!\left( \frac{d \, c^6 \log^4 n}{\min(\alpha,\beta)^3} \right) & \text{if }G=\VarGain \text{, } c=\sup_{y \in \scY} |y|
    \end{array}
\end{align*}
\end{theorem}
\noindent The full proof of Theorem~\ref{thm:fudy_eps_feasible} is deferred to Appendix~\ref{apx:approx}. To get an intuition, suppose $G=\giniGain$. We let $\splfun$ be a max-$G$ decision rule with threshold $\frac{\alpha}{2}$, which we show to be $\Theta(\gamma)$-balanced. Next we set $\epsilon=\Theta(\min(\alpha,\beta))$ so that $\epsilon \le \frac{\gamma}{3}$ and to satisfy the hypotheses of Theorem~\ref{thm:eps_to_beps}. At this point, by Theorem~\ref{thm:main} there is an algorithm that maintains an $\epsilon$-approximate tree using $O\!\left( \frac{d \log^4 n}{\min(\alpha,\beta)^3} \right)$ operations per update request, and by Theorem~\ref{thm:eps_to_beps} the maintained tree is $\beps$-feasible, too.

\paragraph{Remark: adding pruning thresholds.} The original definition of $\beps$-feasibility makes $v$ a leaf also when $|S(T,v)| \le k^*$ or $v$ has depth $h^*$, where $k^*,h^* \in \N$ are given in input. We can include these constraints without altering the bounds of Theorem~\ref{thm:main_h} and Theorem~\ref{thm:main}. For $|S(T,v)| \le k^*$, define $\splfun$ so that $\spl \in \Spl$ if and only if $G(S(T,v)) \ge \epsThr$ \emph{and} $|S(T,v)| > k^*$. For the the depth of $v$, define an ``enriched'' decision function $\hat\splfun$ that takes in input a pair $(S,\zeta)$ where $\zeta \in \N$. One then lets $\hat\splfun(S,\zeta)=\splfun(S)$ if $\zeta>0$, and $\hat\splfun(S,\zeta)$ be the majority/average label of $S$ if $\zeta=0$. Then, $\greedy_{\hat\splfun}$ at $v$ computes $\hat\splfun(S(T,v),\zeta)$ where $\zeta$ equals $h^*$ minus the depth of $v$.

\subsection{Reduction to the delayed approximate decision tree problem}
At the heart of our algorithm \AlgoMain\ lies a reduction to what we we call the \emph{delayed approximate decision tree problem}. Devising an algorithm for that problem is the main technical contribution of our work and is done in Section~\ref{sec:delayed}; in the rest of this section we define the problem and prove how our algorithm relies on it.
\begin{definition}\label{def:delayed_approx}
The \emph{delayed approximate decision tree problem} is as follows. The input is $(\splfun,\epsilon,n,S,U)$ where $\splfun$ is a decision rule, $\epsilon \in (0,1)$, $n \in \N$, $S \in (\scX\times\scY)^n$, and $U$ is a sequence of $\epsilon n$ update requests. Both $S$ and $U$ are given as iterators with $O(d)$ access time per element. If $S'=S+U$ then the output is a decision tree $T$ such that:
\begin{itemize}\itemsep0pt
    \item $T$ is $\epsilon$-approximate w.r.t.\ $(S',\splfun)$
    \item each $v \in V(T)$ keeps an associative array $D(T,v)$ that stores $S'(T,v)$
\end{itemize}
\end{definition}
\noindent The number of operations per update request performed by an algorithm for this problem is the maximum number of elementary operations performed between any two requests pulled from $U$, as well as before pulling the first one and after pulling the last one.

Now suppose we have an algorithm \AlgoRebuild\ for the delayed approximate decision tree problem that performs at most $\tau=\tau(\splfun,\epsilon,n)$ operations per update request. Then \AlgoMain\ maintains a decision tree $T$ as follows. For each vertex $v$ in $T$, \AlgoMain\ stores in an associative array $D(T,v)$ the active set $S_v$ on which the subtree $T_v$ rooted at $v$ was rebuilt the last time (what this means will be clear in a moment). 
For every $v$ \AlgoMain\ also runs an instance of \AlgoRebuild$\big(\splfun,\frac{\epsilon}{2},\frac{\epsilon}{2}|S_v|,S_v,U_v\big)$, denoted by \Instance$(v)$, where the set $S_v$ is fed using an iterator over $D(T,v)$. Now, every time an update request reaches $v$, \AlgoMain\ appends that request to $U_v$ and makes \Instance$(v)$ progress until \Instance$(v)$ pulls that request; by assumption, this happens within $\tau$ operations. In this way, as soon as $\frac{\epsilon}{2} |S_v|$ update requests reach $v$, \Instance$(v)$ returns a pointer to a new tree $T_v'$ that is $\frac{\epsilon}{2}$-approximate w.r.t.\ $S_v+U_v$, which is precisely the current active set at $v$. At that point \AlgoMain\ replaces $T_v$ with $T_v'$ by updating the pointer at $v$'s parent. This keeps the whole tree $\epsilon$-approximate: every $v$ satisfies the definition of $\frac{\epsilon}{2}$-approximation immediately after its subtree is replaced, and afterwards it satisfies the definition of $\epsilon$-approximation until the next $\frac{\epsilon}{2} |S_v|$ update requests reach $v$ --- at which point $T_v$ is replaced again. The pseudocode of \AlgoMain\ is in Algorithm~\ref{algo:main}.
\begin{algorithm}[h!]\small
\caption{\AlgoMain}
\label{algo:main}
\begin{algorithmic}[1]
\State \textbf{Input:} $\splfun$, $\epsilon \in (0,1)$, and a sequence of update requests $\{(s_i,o_i)\}_{i \ge 1}$
\State $T = $ the trivial tree on one leaf $v$ with arbitrary $\lab_v$ \label{main:line:T0} and \Instance$(v)=\nil$
\For{each $i=1,2,\ldots$}\label{main:line:for}
\State compute $P(T,s_i)$
\For{each $v$ in $P(T,s_i)$ in order of increasing distance from the root}\label{main:line:forP}
\If{\Instance$(v)=\nil$}
\State initialize an empty iterator $U(v)$
\State start an instance \Instance$(v)$ of \AlgoRebuild$\big(\splfun,\frac{\epsilon}{2},|D(T,v)|,D(T,v),U(v)\big)$
\EndIf
\State append $(s_i,o_i)$ to $U(v)$
\State \textbf{execute} \Instance$(v)$ until it pulls the next element from $U(v)$  \label{line:main_tau}
\If{\Instance$(v)$ returns a pointer to a tree $T_v'$}
\State set \Instance$(v)=\nil$ and replace $T_v$ with $T_v'$ \label{main:line:replace}
\State \textbf{break} the loop of line~\ref{main:line:forP}
\EndIf
\EndFor
\EndFor
\end{algorithmic}
\end{algorithm}

\subsection{Guarantees}
The next two lemmas bind the guarantees of \AlgoMain\ to those of \AlgoRebuild. Together with Theorem~\ref{thm:delayed_full} (Section~\ref{sec:delayed}) they give Theorem~\ref{thm:main_h}. Let $T^0$ be the tree computed at line~\ref{main:line:T0}, and for all $i \ge 1$ let $T^i$ be the tree $T$ at the end of the $i$-th iteration of the loop at line~\ref{main:line:for}. 

\begin{lemma}\label{lem:main_correct}
If \AlgoRebuild\ is an algorithm for the delayed approximate decision tree problem, then $T^i$ is $\epsilon$-approximate w.r.t.\ $(S^i,\splfun)$ for all $i \ge 0$.
\end{lemma}
\begin{proof}
First, we prove that the arrays $D(T,v)$ are updated correctly. For each $v \in V(T^i)$ let $i_v \le i$ be the most recent iteration where $v$ is in a subtree created by \AlgoRebuild, or $i_v=0$ if no such iteration exists. We claim that, for all $i \ge 1$ and all $z \in V(T^i)$, $D(T^i,z)$ stores precisely $S^{i_z}(T^i,z)$. This is trivial for $i=0$, so suppose the claim holds for $T^{i-1}$ and let $z \in V(T^i)$. If $i_z < i$ then the claim holds for $z$ since $D(T^i,z)=D(T^{i-1},z)$. Otherwise $T^i_z$ is a subtree of a tree $T^i_v$ that has been returned by \AlgoRebuild\ at iteration $i$, and by induction \AlgoRebuild\ has been given in input $D(T^{i-1},v)$ and all subsequent updates that reached $v$. By construction of \AlgoRebuild\ this implies that $D(T^i,z)$ stores precisely $S^{i_z}(T^i,z)$.

Now consider a generic iteration of \AlgoMain\ where line~\ref{main:line:replace} is executed. By the claim above and by construction of \AlgoMain, $\Instance(v)$ was given in input a set of examples and a sequence of updates whose corresponding active set is precisely $S^i(T^i,v)$. Thus $T^i_v$ is $\frac{\epsilon}{2}$-approximate w.r.t.\ $(S^i(T^i,v),\splfun)$. Now let $j > i$. If less than $\frac{\epsilon}{2} |S^i(T^i,v)|$ requests have reached $v$ between iteration $i+1$ and the end of iteration $j$, then $T^j_v$ is $\epsilon$-approximate w.r.t.\ $(S^j(T^j,v),\splfun)$; and before the $\frac{\epsilon}{2} |S^i(T^i,v)|$-th such request reaches $v$, the subtree $T^i_v$ will be replaced, making it $\frac{\epsilon}{2}$-approximate again.
\end{proof}

\begin{lemma}\label{lem:main:time_blackbox}
If \AlgoRebuild\ performs at most $\tau$ operations per update request, then \AlgoMain\ performs $O\!\left(\tau \cdot h\right)$ operations per update request where $h=\max_{i \ge 1} h(T^i)$.
\end{lemma}
\begin{proof}
Straightforward.
\end{proof}

\section{Delayed construction of an approximate decision tree}
\label{sec:delayed}
Recall Definition~\ref{def:delayed_approx}, and let $n=|S|$ and $\epsilon n = |U|$. This section proves:
\begin{theorem}\label{thm:delayed_full}
There is an algorithm \AlgoRebuild\ for the delayed approximate decision tree problem that uses a number of operations per update request in
$$O\!\left( h \, \log n \cdot \left(h + \frac{d \log n}{\epsilon} + \frac{\cost_{\splfun}(m)}{\epsilon n}\right) \right)$$
where $m \le n(1+\epsilon)$ is the maximum size of the active set obtained from applying any prefix of $U$ to $S$ and $h \le h_{\splfun}(\epsilon,m)$ is the maximum height of the tree held by the algorithm.
\end{theorem}
\noindent Recalling the examples of Section~\ref{sec:dyn}, if $h_{\splfun}(\epsilon,n)=O\!\left(\frac{\log n}{\epsilon}\right)$ and $\cost_{\splfun}(m)=O(d\, m \log n)$ the bound of Theorem~\ref{thm:delayed_full} is in $O\!\left(\frac{d \log^3 n}{\epsilon^2}\right)$. 
The rest of this section describes \AlgoRebuild\ (Algorithm~\ref{algo:rebuild}) and proves two theorems, Theorem~\ref{thm:rebuild_correct} and Theorem~\ref{thm:delayed_time}, that yield immediately Theorem~\ref{thm:delayed_full}.

\subsection{Overview of the algorithm}
\AlgoRebuild\ works in rounds, and at each round it halves the number of requests it pulls from $U$. Without loss of generality assume $\epsilon n = 2^{\ell-1}$ for some $\ell \in \N$ (otherwise just replace $\epsilon$ with an appropriate $\epsilon' \in (\frac{\epsilon}{2},\epsilon)$), and define:
\begin{align}
    t_i &= \left\{
        \begin{array}{ll}
             \epsilon n(1-2^{-i}) \quad & i=0,\ldots,\ell-1  \\
             \epsilon n & i \ge \ell
        \end{array}
    \right.
\end{align}
For all $i,j \ge 0$ define $U_i^j = U[t_i+1,\ldots,t_j]$; and to simplify the notation let $U_i=U_i^{\ell}$ and $U^j=U_0^{j}$. Note that, for all $i = 1,\ldots,\ell-1$, the sequence $U^i$ contains all but the last $2^{-i}|U|$ requests of $U$, while $U^{\ell}=U$. Finally, for all $i \ge 0$ let $S^i = S+U^i$. Thus $S^i$ is the active set obtained by applying to $S$ all but the last $2^{-i}|U|$ requests of $U$; in particular, $S^0=S$ and $S^{\ell}=S+U$.

Let us give the intuition of \AlgoRebuild, assuming for simplicity $\cost_{\splfun}(n)= O(d\,n\log n )$, in which case $\greedy_{\splfun}$ runs in time $O(d\, h\, n \log n)$ on an $n$-element input --- see Lemma~\ref{lem:greedy}. To begin, while the requests in $U^1$ are pulled, compute $T^0=\greedy_{\splfun}(S)$. Since this takes $O(d\,h\, n \log n)$ operations and $|U^1|=\frac{\epsilon n}{2}$, we use $O\!\left(\frac{d\, h\, \log n}{\epsilon}\right)$ operations per request.
Now suppose that, for some $i \ge 1$, before pulling any request of $U^{i+1}$ we know $T^{i-1}=\greedy_{\splfun}(S^{i-1})$; by the argument above this holds for $i=1$. By the properties of $\greedy_{\splfun}$, every $v\in V(T^{i-1})$ keeps an associative array $D(T^{i-1},v)$ which contains $S^{i-1}(T^{i-1},v)$. While the requests of $U_i^{i+1}$ are pulled, update the associative arrays of $T^{i-1}$ using the requests of $U_{i-1}^i$, and mark every vertex $v$ that is reached by more than $\epsilon |D(T^{i-1},v)|$ of those requests --- those are the $v$ that could violate the $\epsilon$-approximation.
Finally, take all marked vertices and rebuild their subtrees with $\greedy_{\splfun}$, obtaining $T^i$. Since each request in $U_{i-1}^i$ reaches at most $h$ vertices, and marked vertices have sets of size at most $\epsilon$ times the number of the requests that reached them, the total size of the sets at marked vertices is at most $\frac{h \, |U_{i-1}^i|}{\epsilon} = h n 2^{-i}$.
Thus, rebuilding the subtrees requires $O(d \, h^2\, 2^{-i} n \log n)$ operations in total, which, since $|U_i^{i+1}| = \epsilon n 2^{-(i+1)}$, means $O\!\left(\frac{d\, h^2\, \log n}{\epsilon}\right)$ operations per request. Unfortunately this does not work yet, since we did not keep the invariant $T^{i}=\greedy_{\splfun}(S^i)$. However, by slightly decreasing the approximation parameter to $\frac{\epsilon}{\lg_2 n}$ one can guarantee that if $T^{i-1}$ is $\epsilon$-approximate w.r.t.\ $(S^{i-1},\splfun)$ then $T^{i}$ is $\epsilon$-approximate w.r.t.\ $(S^{i},\splfun)$. The final bound is obtained by expliciting the dependence on $\cost_{\splfun}$, refining the analysis to turn the $h^2$ into an $h$, and taking into account the cost of maintaining some ancillary data structures.

\begin{algorithm}[h!]\small
\caption{\AlgoRebuild}
\label{algo:rebuild}
\begin{algorithmic}[1]
\State \textbf{Input:} decision rule $\splfun$, $\epsilon \in (0,1)$, $n \in \N$, set $S$ of examples, sequence $U$ of $\epsilon |S|$ update requests
\State let $\ell=\lg_2(\epsilon n)+1$
\State let $\tau$ be as in Theorem~\ref{thm:delayed_full}
\State \textbf{execute} what follows, using at most $\tau$ operations for each request pulled from in $U_0^1$:
\State\hspace{\algorithmicindent} $T=\greedy_{\splfun}(S)$  \label{line:T0}
\For{$i=1,\ldots,\ell$}\label{line:for_so}
\textbf{execute} what follows, using at most $\tau$ operations for each request pulled from $U_i^{i+1}$:
\State initialize a set $\Vr_i = \emptyset$
\For{each $(s,o) \in U_{i-1}^i$} \label{line:update_loop}
\State compute the path $P(T,s)$ of $s$ in $T$
\For{each vertex $v \in P(T,s)$} \label{line:update_path}
\State update $D(T,v)$ with $(s,o)$, and increment $\Delta_i(v)$ \label{line:update}
\If{$\Delta_i(v) \ge \frac{\epsilon}{4 \lg_2 n} \, n_{i-1}(v)$}\label{line:cond_b}
\State $\Vr_i = \Vr_i \cup \{v\}$
\EndIf
\EndFor
\EndFor
\State let $\Vr_i^*$ be the subset of maximal\footnotemark\ vertices in $\Vr_i$\label{line:clean_Vi}
\For{every $v \in \Vr_i^*$} \label{line:rebuild_for}
\State $T_v=\greedy_{\splfun}(D(T,v))$ \label{line:rebuild}
\EndFor \label{line:rebuild_for_end}
\EndFor
\State\Return a pointer to $T$
\end{algorithmic}
\end{algorithm}
\footnotetext{Having no proper ancestor in $\Vr_i$.}

\subsection{Setup of the analysis}
Let us pin down some necessary notation. For all $i=0,\ldots,\ell$, we call \emph{round $i$} the $i$-th execution of the \textbf{for} loop (for $i=0$ this is undefined), and denote by $T^i$ the tree $T$ held by \AlgoRebuild\ just before the beginning of round $i+1$.
We let $h=\max_{i\in\{0,\ldots,\ell\}}h(T^i)$.

\AlgoRebuild\ keeps track of two counters, $n_i(v)$ and $\Delta_i(v)$. Let us define them formally; we show below how to update them efficiently. Let $i \in \{0,\ldots,\ell\}$.  For every $v \in V(T^i)$:
\begin{align}
    n_i(v) = |S^i(T^i,v)|
\end{align}
Note that, while $T^i$ is the tree held just \emph{after} round $i$, the set $S^i$ takes into account only the updates received \emph{before} round $i$; thus $S^i(T^i,v)$ is \emph{not} the active set at $v$ in $T^i$ at the end of round $i$ (that would be $S^{i+1}(T^i,v)$). 
Moreover, for every $v \in V(T^i)$:
\begin{align}
    \Delta_i(v) = \left| \{(s,o) \in U_{i-1}^i \;:\; v \in P(T^i,x) \} \right|
\end{align}
In words, $\Delta_i(v)$ is the number of requests of $U_{i-1}^i$ that reach $v$.
Our analysis also needs the number of requests that have reached $v$ since $v$ was created. To this end, define:
\begin{align}
    c_i(v) &= \left\{
    \begin{array}{ll}
         0 & \text{ if $i=0$ or $\Vr_i^*$ contains $v$ or an ancestor of $v$} \\
         c_{i-1}(v)+\Delta_i(v) & \text{ otherwise}
    \end{array}
    \right.
\end{align}
Note that:
\begin{align}
    c_i(v)-c_{i-1}(v) \le \Delta_i(v), \qquad \forall i \ge 1 \label{eq:ci_Di}    
\end{align}
Our proofs use the following facts.
\begin{lemma}\label{lem:path}
If $v \in \Vr_i^*$ then $S^i(T^i,v) = S^i(T^{i-1},v)$.
\end{lemma}
\begin{proof}
By construction $\Vr_i^*$ contains no proper ancestor of $v$, hence all such ancestors are untouched by the loop of line~\ref{line:rebuild_for}.
\end{proof}
\begin{lemma}\label{lem:n_i}
For every $i=1,\ldots,\ell$ every $v \in \Vr_i^*$ satisfies $n_i(v) \le n_{i-1}(v) + \Delta_i(v)$.
\end{lemma}
\begin{proof}
By Lemma~\ref{lem:path} $|S^i(T^i,v)|=|S^i(T^{i-1},v)|$, hence:
\begin{align}
    n_{i}(v) &= |S^i(T^i,v)| \label{eq:n_Ti}
    \\&= |S^i(T^{i-1},v)| \label{eq:n_Ti1}
    \\&\le |S^{i-1}(T^{i-1},v)| + \Delta_i(v) \label{eq:n_Ti2}
    \\&= n_{i-1}(v) + \Delta_i(v) \label{eq:n_Ti3}
\end{align}
where~\eqref{eq:n_Ti2} is straightforward and~\eqref{eq:n_Ti},\eqref{eq:n_Ti3} use the definition of $n_i$.
\end{proof}

The next two subsections prove respectively Theorem~\ref{thm:rebuild_correct} and Theorem~\ref{thm:delayed_time}, which form the two parts of Theorem~\ref{thm:delayed_full} (correctness of \AlgoRebuild\ and bound on the number of operations per update request).
The correctness is given directly by Theorem~\ref{thm:rebuild_correct}. For the number of operations, Theorem~\ref{thm:delayed_time} yields a bound of
$O\!\left( h \, \log n \cdot \left(h + \frac{d \log n}{\epsilon} + \frac{\cost_{\splfun}(m)}{\epsilon n}\right) \right)$ where $h = \max_{i=0,\ldots,\ell}h(T_i)$ and $m=\max_{i=0,\ldots,\ell}|S^i|$. Since by Theorem~\ref{thm:rebuild_correct} every $T_i$ is $\epsilon$-approximate w.r.t. $(S^i,\splfun)$ and $|S^i| \le m$, then by definition $h(T_i) \le h_{\splfun}(\epsilon,m)$ for all $i$. Moreover clearly $m \le |S|+|U| = (1+\epsilon)$.

\subsection{Correctness of \AlgoRebuild}
\begin{lemma}\label{lem:rebuild_correct}
Let $i \in \{0,\ldots,\ell\}$. Then every $v \in V(T^i)$ satisfies $c_i(v) \le \epsilon\, n_{i_v}(v)$ where $i_v \in \{0,\ldots,i\}$ is the last round where $v$ was created.
\end{lemma}
\begin{proof}
The proof is trivial for $i=0$ since $c_0(v)=0$. Let then $i\ge 1$. Since $v$ was built at round $i_v$ then $c_{i_v}(v)=0$; using a telescoping sum and applying~\eqref{eq:ci_Di},
\begin{align}
    c_{i}(v) = c_{i}(v) - c_{i_v}(v) = \sum_{j=i_v+1}^{i} \left(c_{j}(v) - c_{j-1}(v)\right) \le \sum_{j=i_v+1}^{i} \Delta_j(v) \label{eq:civ_sum}
\end{align}
We shall then bound $\Delta_j(v)$. By definition of $i_v$, for all $j \in \{i_v+1,\ldots,i\}$ at round $j$ the condition of line~\ref{line:cond_b} fails, hence:
\begin{align}
    \Delta_j(v) < \frac{\epsilon}{4 \lg_2 n} n_{j-1}(v) \label{eq:c_ti}
\end{align}
In this case, since $n_{j}(v) \le n_{j-1}(v) + \Delta_j(v)$ by Lemma~\ref{lem:n_i}, we have:
\begin{align}
    n_{j}(v) \le n_{j-1}(v)\left(1+ \frac{\epsilon}{4 \lg_2 n}\right) \label{eq:n_ti}
\end{align}
By iterating~\eqref{eq:n_ti} we conclude that for every $j=i_v+1,\ldots,i$:
\begin{align}
    n_{j}(v) \le \left(1+ \frac{\epsilon}{4 \lg_2 n}\right)^{j-i_v} n_{i_v}(v)
\end{align}
and~\eqref{eq:c_ti} then implies for every $j=i_v+1,\ldots,i$:
\begin{align}
    \Delta_j(v) \le n_{i_v}(v)\cdot \frac{\epsilon}{4 \lg_2 n} \left(1+ \frac{\epsilon}{4 \lg_2 n}\right)^{j-1-i_v} 
\end{align}
Plugging this bound in~\eqref{eq:civ_sum} and noting that $i \le 2 \lg_2 (\epsilon n) \le 2 \lg_2 n$, we obtain:
\begin{align}
    c_{i}(v) &\le n_{i_v}(v) \cdot \frac{\epsilon}{4 \lg_2 n}  \sum_{j=i_v+1}^{i} \left(1+\frac{\epsilon}{4 \lg_2 n}\right)^{j-1-i_v} 
    \\ &\le n_{i_v}(v) \cdot \frac{\epsilon i}{4 \lg_2 n} \left(1+ \frac{\epsilon}{4 \lg_2 n}\right)^{i} 
    \\ &\le n_{i_v}(v) \cdot \frac{\epsilon}{2} \left(1+ \frac{\epsilon}{4 \lg_2 n}\right)^{2 \lg_2 n}
    \\ &< n_{i_v}(v) \cdot \epsilon\, \frac{e^{1/2}}{2}
\end{align}
which is at most $\epsilon \cdot n_{i_v}(v)$, as claimed.
\end{proof}

\begin{theorem}\label{thm:rebuild_correct}
$T^i$ satisfies the constraints of the delayed approximate decision problem w.r.t.\ $S^i,\splfun$ for all $i \in \{0,\ldots,\ell\}$. Hence, in particular, the tree returned by \AlgoRebuild\ is $\epsilon$-approximate w.r.t.\ $(S+U,\splfun)$.
\end{theorem}
\begin{proof}
First, we show that $T^i$ is $\epsilon$-approximate w.r.t.\ $(S^i,\splfun)$ for all $i \in \{0,\ldots,\ell\}$. For $i=0$ the claim is trivial since $T^0=\greedy_{\splfun}(S^0)$. Let then $i \ge 1$, let $v \in V(T^i)$, and let $i_v \in \{0,\ldots,i\}$ be the last round where $v$ was created (i.e., where $v$ or some ancestor of $v$ was in $\Vr_{i_v}^*$).
By construction, the split rule at $v$ in $T^i$ is $\splfun(S^{i_v}(T^{i_v},v))$; and by Lemma~\ref{lem:rebuild_correct} $c_i(v) \le \epsilon \, |S^{i_v}(T^{i_v},v)|$, so $\EDR(S^i(T^i,v),S^{i_v}(T^{i_v},v)) \le \epsilon$. Thus $T^i$ is $\epsilon$-approximate w.r.t.\ $(S^i,\splfun)$, as claimed.

Next, we show that for every $v \in V(T^i)$ the array $D(T^i,v)$ stores exactly $S^i(T^i,v)$. This is true for $i=0$ by definition of $\greedy_{\splfun}$. Now let $i \ge 1$ and suppose the claim is true for $i-1$. Because of line~\ref{line:update}, at the end of the loop of line~\ref{line:for_so} each $v \in V(T^{i-1})$ satisfies that $D(T^{i-1},v)$ stores $S^i(T^{i-1},v)$. By definition of $\greedy_{\splfun}$, then, after the rebuilds at line~\ref{line:rebuild} $D(T^i,v)$ stores $S^i(T^i,v)$ for all $v \in V(T^i)$, as claimed.

For the second claim just note that $S+U=S^{\ell}$ and \AlgoRebuild\ returns $T^{\ell}$.
\end{proof}

\subsection{Performance of \AlgoRebuild}
\begin{lemma}\label{lem:counters}
\AlgoRebuild\ can be implemented so that, at every round $i \ge 1$:
\begin{enumerate}\itemsep0pt
    \item each iteration of the loop at line~\ref{line:update_loop} takes $O(d \, h \log n)$ operations
    \item line~\ref{line:clean_Vi} takes $O(|\Vr_i|\, h \log n)$ operations
    \item line~\ref{line:rebuild_for} can enumerate $\Vr_i^*$ in $O(1)$ per element
\end{enumerate}
\end{lemma}
\begin{proof}
1. By definition $|P(T,s)| \le h$, hence computing $P(T,s)$ takes time $O(h)$. For each $v \in P(T,s)$, updating $D(T,v)$ takes time $O(d \log n)$ by assumption. To increment $\Delta_i(v)$ in time $O(1)$, create it as a new variable associated to $v$ the first time $v$ is processed by the loop of line~\ref{line:update_path} and set it to $1$, then mark $v$ as ``alive'' so that subsequent updates increment that variable. Assuming we can access $n_{i-1}(v)$ in time $O(1)$, checking the condition at line~\ref{line:cond_b} takes time $O(1)$. Finally, updating $\Vr_i$ takes time $O(\log |V(T)|)=O(\log n)$ using an associative array with logarithmic update time. It remains to show how line~\ref{line:cond_b} can  access $n_{i-1}(v)$ in $O(1)$ operations.

Consider again $D(T,v)$. Just before round $i+1$ starts, $D(T,v)$ stores $S^i(T^i,v)$. This is true for $i=0$ since $\greedy_{\splfun}$ stores explicitly $S^0(T^0,v)$ in $D(T,v)$; and it remains true for $i \ge 1$ since either $v$ is in a subtree rebuilt at round $i$, and the argument above applies, or $D(T,v)$ is updated by line~\ref{line:update}. Thus, for each $i \ge 1$, at the beginning of round $i$ we can access $n_{i-1}(v)$ in time $O(1)$ by querying $|D(T,v)|$. To make it available throughout all the round, right before executing line~\ref{line:update} query $|D(T,v)|$ and store it in a new variable $\hat n_{i-1}(v)$, then mark $v$ as ``done'' so that $\hat n_{i-1}(v)$ does not get overwritten. From this point onward, for any $v \in V(T)$ one can retrieve $n_{i-1}(v)$ in time $O(1)$ by using $\hat n_{i-1}(v)$ if it exists, and using $|D(T,v)|$ otherwise.

2,3. Initialise an empty linked list $\Vr_i^*=\emptyset$. For every $v \in \Vr_i$, list all the ancestors of $v$ in $T$ --- this takes time $O(h)$ as every vertex of $T$ keeps a pointer to its parent --- and if none of them is in $\Vr_i$ then append $v$ to $\Vr_i^*$.
\end{proof}

\begin{theorem}\label{thm:delayed_time}
\AlgoRebuild\ can be implemented to use $O\!\left( h \, \log n \cdot \left(h + \frac{d \log n}{\epsilon} + \frac{\cost_{\splfun}(m)}{\epsilon n}\right) \right)$ operations per update request where $h = \max_{i=0,\ldots,\ell}h(T_i)$ and $m=\max_{i=0,\ldots,\ell}|S^i|$.
\end{theorem}
\begin{proof}
By Lemma~\ref{lem:greedy} and by definition of $h$, $\greedy_{\splfun}(S)$ runs in time:
\begin{align}
    O\!\left(h \cdot \Big(\cost_{\splfun}(n) + d\, n \log n \Big)\right)
\end{align}
Thus, $\greedy_{\splfun}(S)$ can be ran by using for each request in $U_0^1$ a number of operations in:
\begin{align}
    O\!\left(\frac{h}{\epsilon} \cdot \left(\frac{\cost_{\splfun}(n)}{n} + d\, \log n \right)\right) \label{eq:ops_R0}
\end{align}

Now consider round $i$. By Lemma~\ref{lem:counters}, the total number of operations taken by the loop at line~\ref{line:update_loop} together with line~\ref{line:clean_Vi} is in $O(t_i \, d h \log n + |\Vr_i| h \log n)$. As each iteration of that loop inserts at most $h$ elements in $\Vr_i$ then $|\Vr_i| \le t_i \, h$, so the bound above is in $O(t_i \cdot (d + h) h \log n)$. Hence, excluding line~\ref{line:rebuild_for}, the $i$-th round can be ran using $O((d + h) h \log n)$ operations per update request.
It remains to bound the time taken by the loop at line~\ref{line:rebuild_for}, which is dominated by the total time of the invocations of $\greedy_{\splfun}$. Let then $v \in \Vr_i^*$ and consider $D(T,v)$. By construction, $D(T,v)$ stores $S^i(T^{i-1}(v))$, which by Lemma~\ref{lem:path} equals $S^i(T^i,v)$.
Thus, by Lemma~\ref{lem:greedy}, by definition of $h$, and by the assumptions on $\cost_{\splfun}$ the total time of the invocations of $\greedy_{\splfun}$ is in:
\begin{align}
    & O\!\left(\sum_{v \in \Vr_i^*} h \cdot \Big(\cost_{\splfun}(n_i(v)) + d\,n_i(v)\log n_i(v) \Big)\right)
    \\ &=
    O\!\left(h \,  \cost_{\splfun}\left(\sum_{v \in \Vr_i^*}n_i(v)\right)\right)  +  O\!\left(d h \log n \, \sum_{v \in \Vr_i^*} n_i(v) \right) \label{eq:Vri_tot_cost}
\end{align}
Thus, we shall bound $\sum_{v \in \Vr_i^*} n_i(v)$. Fix any $v \in \Vr_i^*$. By line~\ref{line:cond_b}:
\begin{align}
    n_{i-1}(v) \le \frac{4 \lg_2 n}{\epsilon} \Delta_i(v)
\end{align}
Moreover $n_i(v) \le n_{i-1}(v) + \Delta_i(v)$ by Lemma~\ref{lem:n_i}, hence:
\begin{align}
    n_{i}(v) \le \left(1+\frac{4 \lg_2 n}{\epsilon}\right)  \Delta_i(v) \le \frac{5 \lg_2 n}{\epsilon}  \Delta_i(v)
\end{align}
Since no two vertices in $\Vr_i^*$ are in an ancestor-descendant relationship, every $(s,o) \in U_{i-1}^i$ reaches at most one vertex in $\Vr_i^*$. Therefore:
\begin{align}
    \sum_{v \in \Vr_i^*} \Delta_i(v) \le |U_{i-1}^i| = t_{i-1} - t_i = t_i
\end{align}
We conclude that:
\begin{align}
    \sum_{v \in \Vr_i^*} n_i(v) \le \frac{5 \lg_2 n}{\epsilon} \cdot t_i \label{eq:sum_Vr}
\end{align}
We can now bound the two terms of~\eqref{eq:Vri_tot_cost}. For the first term, we consider two cases. If $\frac{t_i \, 5 \log_2 n}{\epsilon n} > 1$ then we use again the fact that no two vertices in $\Vr_i^*$ are in an ancestor-descendant relationship to obtain $\sum_{v \in \Vr_i^*} n_i(v) \le |S^i| \le m$, which yields:
\begin{align}
\cost_{\splfun}\left(\sum_{v \in \Vr_i^*}n_i(v)\right) \le \cost_{\splfun}(m) < \frac{\cost_{\splfun}(m) \, t_i \, 5 \log_2 n}{\epsilon n}
\end{align}
If instead $t_i \le \frac{\epsilon n}{5 \log_2 n}$ then observe that, by its assumptions, $\cost_{\splfun}$ satisfies $\cost_{\splfun}(x) \le \frac{\cost_{\splfun}(cx)}{c}$ for all $x \in \R_{\ge 0}$ and all $c \ge 1$. Using~\eqref{eq:sum_Vr} and choosing $c=\frac{\epsilon n}{t_i 5 \log_2 n} < 1$, we obtain:
\begin{align}
\cost_{\splfun}\left(\sum_{v \in \Vr_i^*}n_i(v)\right) \le 
\cost_{\splfun}\left(\frac{5 \lg_2 n}{\epsilon} \cdot t_i\right)
\le \frac{\cost_{\splfun}(n) \, t_i \,5 \log_2 n}{\epsilon n}
\le \frac{\cost_{\splfun}(m) \, t_i \,5 \log_2 n}{\epsilon n}
\end{align}
Therefore the first term of~\eqref{eq:Vri_tot_cost} is bounded by:
\begin{align}
    O\!\left(t_i \cdot \frac{h \, \cost_{\splfun}(m) \, \log_2 n}{\epsilon n} \right)
\end{align}
For the second term of~\eqref{eq:Vri_tot_cost}, again by~\eqref{eq:sum_Vr} we obtain:
\begin{align}
    d h \log n \, \sum_{v \in \Vr_i^*} n_i(v)
    = O\left( t_i \cdot \frac{h\, d \log^2 n }{\epsilon} \right)
\end{align}
Thus the total number of operations performed by $\greedy_{\splfun}$ in the loop of line~\ref{line:rebuild_for} is in:
\begin{align}
    O\!\left( t_i \cdot \frac{h\, \log n}{\epsilon}\, \left( \frac{\cost_{\splfun}(2 n)}{n} + d \log n \right) \right) 
\end{align}
Summing all bounds and dividing by $t_i$, one obtains the following bound on the number of operations per update request for each round $i=1,\ldots,\ell$:
\begin{align}
    &~~~~ O\!\left((d + h) h \log n + \frac{h \, \log n}{\epsilon} \left( \frac{\cost_{\splfun}(m)}{n} + d \log n \right) \right)  \label{eq:ops_Ri}
\end{align}
Since~\eqref{eq:ops_Ri} dominates~\eqref{eq:ops_R0}, then it bounds the operations per request of \AlgoRebuild.
By rearranging terms, we obtain that~\eqref{eq:ops_Ri} is bounded by:
\begin{align}
    O\!\left( h \, \log n \cdot \left(h + \frac{d \log n}{\epsilon} + \frac{\cost_{\splfun}(m)}{\epsilon n}\right) \right)      \label{eq:ops_delayed}
\end{align}
which concludes the proof.
\end{proof}

\section{Smoothness, approximation, and balancedness of gains}\label{sec:smooth}
In this section we prove that $G\in\{\giniGain,\IG,\VarGain\}$ is smooth w.r.t.\ the relative edit distance; that is, that for any split rule $\spl$ we can bound $|G(S,\spl)-G(S',\spl)|$ in term of $\EDR(S,S')$. As a consequence we also prove that an $\epsilon$-approximate tree also guarantees a good approximation in terms of $G$, and that max-$G$ $\alpha$-threshold decision rules are $\gamma$-balanced where $\alpha$ depends on $\gamma$. We first prove some ancillary results on the functions $g$ for which $G$ is a conditional $g$-gain, and then move on to prove the rest.

\subsection{Ancillary results}
\begin{lemma}\label{lem:EDR_g}
    Let $g : (\scX \times \scY)^* \to \R_{\ge 0}$ and let $f : \N \to \R$ be nondecreasing and such that:
    \begin{align}
        g(S) &\le f(|S|) && \forall \, S \in (\scX \times \scY)^*
        \\
        |g(S)-g(S')| &\le \frac{f(\max(|S|,|S'|))}{\max(|S|,|S'|)} && \forall \, S,S' \in (\scX \times \scY)^*
    \end{align}
    Then for all $S,S' \in (\scX \times \scY)^*$:
    \begin{align}
        |g(S)-g(S')| \le 
        3 \EDR(S,S') f(\max(|S|,|S'|))
    \end{align}
\end{lemma}
\begin{proof}
For simplicity let $k=\ED(S,S')$. If $k = 0$ then $g(S)=g(S')$ and the bound is trivial, so assume $k \ge 1$ and $|S'|\ge|S|$. 
Suppose first $\EDR(S,S') \ge \frac{1}{3}$. Then:
\begin{align}
    |g(S)-g(S')| \le \max(g(S),g(S')) \le f(\max(|S|,|S'|))
    \le 3 \EDR(S,S') f(\max(|S|,|S'|)) 
\end{align}
Now suppose instead $\EDR(S,S') < \frac{1}{3}$. Observe that this implies $|S'| \le \frac{3}{2}|S|$ and $|S|\ge 2$, and thus $|S|-1 \ge \frac{1}{3}|S'|$.
By definition of $\ED$ there exist $S_0,\ldots,S_k \in (\scX \times \scY)^*$ with $S_0=S,S_k=S'$, and such that $\ED(S_i,S_{i+1})=1$ for all $i=0,\ldots,k-1$. Note that in particular there exists such a set where $|S|-1 \le |S_i| \le |S'|$ for all $i$. By the properties of $f$ this implies:
\begin{align}
    |g(S)-g(S')| &\le \sum_{i=0}^{k-1} |g(S_i)-g(S_{i+1})|
    \\ &\le \sum_{i=0}^{k-1} \frac{f(\max(|S_i|,|S_{i+1}|))}{\max(|S_i|,|S_{i+1}|)} 
    \\ &\le \sum_{i=0}^{k-1} \frac{f(|S'|)}{|S|-1} 
    \\ &= k \frac{f(|S'|)}{|S|-1} 
    \\ &\le 3 k \frac{f(\max(|S|,|S'|))}{\max(|S|,|S'|)}
\end{align}
which equals $3 \EDR(S,S') f(\max(|S|,|S'|))$.
\end{proof}

\subsubsection{Gini impurity}
\begin{lemma}\label{lem:ED_gini_imp}
    If $\ED(S,S')\le 1$ then $|\giniImp(S)-\giniImp(S')|\le\frac{4}{\max(|S|,|S'|)}$.
\end{lemma}
\begin{proof}
    The claim is trivial if $S=S'$, so assume $S'=S+s$. Without loss of generality we may assume $s=(x,y)$ where $y=1$. Clearly $\giniImp(S)-\giniImp(S') = \sum_{i=1}^k (p_i^2(S')-p_i^2(S))$, and since $p_i^2(S')-p_i^2(S)=(p_i(S')+p_i(S))(p_i(S')-p_i(S))$ and $(p_i(S')+p_i(S))\le 2$, then:
    \begin{align}
        |\giniImp(S)-\giniImp(S')| &= \left|p_1^2(S')-p_1^2(S) + \sum_{i=2}^k (p_i^2(S')-p_i^2(S))\right|
        \\ &\le 2|p_1(S')-p_1(S)| + 2\left|\sum_{i=2}^k (p_i(S')-p_i(S))\right|
    \end{align}
    Standard calculations give:
    \begin{align}
        p_i(S')-p_i(S) = \left\{
        \begin{array}{ll}
            \frac{1-p_i(S)}{n+1} & i=1 \\
            -\frac{p_i(S)}{n+1} & i \ge 2
        \end{array}
        \right.
    \end{align}
    Thus
    \begin{align}
        2|p_1(S')-p_1(S)| + 2\left|\sum_{i=2}^k (p_i(S')-p_i(S))\right|
        &= 2\frac{1-p_1(S)}{n+1} + 2 \sum_{i=2}^k\frac{p_i(S)}{n+1}
        \\ &= 4\frac{1-p_1(S)}{n+1}
    \end{align}
    which is at most $\frac{4}{n+1}=\frac{4}{\max(|S|,|S'|)}$, as claimed.
\end{proof}

\subsubsection{Entropy} 
\begin{claim}\label{claim:H_p}
$H(p)\le 3 p \log \frac{1}{p}$ for all $p \in [0,\nicefrac{1}{2}]$.
\end{claim}
\begin{proof}
\cite[Theorem 1.1]{Topsoe2001} and easy manipulations yield:
\begin{align}
    H(p) \le \frac{\log p \, \log(1-p)}{\log 2} = \frac{\log \frac{1}{p} \log\left(1+\frac{p}{1-p}\right)}{\log 2} \le \frac{\log \frac{1}{p}}{\log 2} \frac{p}{1-p} \le 3 p \log \frac{1}{p}
\end{align}
\end{proof}

\begin{claim}\label{claim:H_tvd}
For any two random variables $X,X'$ defined on the same space of events and taking on at most $n$ distinct values:
\begin{align}
    |H(X) - H(X')| \le \tvd{X}{X'} \lg_2(n-1) + H(\tvd{X}{X'})
\end{align}
\end{claim}
\begin{proof}
This is a special case of the Fannes–Audenaert inequality~\cite{Audenaert_2007} for diagonal matrices.
\end{proof}

\begin{lemma}~\label{lem:ED_entropy}
If $\ED(S,S') \le 1$ then $|H(S)-H(S')| < \frac{5\log n}{n}$ where $n=\max(|S|,|S'|)$.
\end{lemma}
\begin{proof}
    The claim is trivial if $S=S'$, hence assume $S'=S+s$, and let $(X,Y)$ be a uniform random element of $S$ and $(X',Y')$ a uniform random element of $S'$. Clearly $\tvd{Y}{Y'} = \frac{1}{n}$, hence by the definition of $H(S),H(S')$, and by Claim~\ref{claim:H_tvd} and Claim~\ref{claim:H_p}, and since $n\ge 2$ and so $\frac{1}{n}\le \frac{1}{2}$:
    \begin{align}
        |H(S) - H(S')| &\le \frac{1}{n} \lg_2(n-1) + H\left(\frac{1}{n}\right)
        < \frac{2}{n} \log n + \frac{3}{n} \log n
    \end{align}
    concluding the proof.
\end{proof}

\subsubsection{Variance}
Let $c = \sum_{y \in \scY} \frac{c}{2}$ for some $c \in \R_{\ge 0}$. Clearly this implies $\Var(S) \le c^2$ for all $S$.
\begin{lemma}\label{lem:ED_variance}
    If $\ED(S,S')\le 1$ then $|\Var(S)-\Var(S')|\le\frac{3 c^2}{\max(|S|,|S'|)}$.
\end{lemma}
\begin{proof}
The claim is trivial if $S=S'$, so assume $S'=S+s$ for some $s=(x_s,y_s)$ and let $n=|S'|$. Standard calculations show that:
\begin{align}
    \Var(S+s) &= \Var(S) \frac{(n-1)^2}{n^2} + \frac{1}{n^2} \sum_{(x,y) \in S}(y-y_s)^2
    \\
    &= \Var(S)  -\frac{2n-1}{n^2}\Var(S) + \frac{1}{n^2} \sum_{(x,y) \in S}(y-y_s)^2
\end{align}
Thus:
\begin{align}
    |\Var(S+s)-\Var(S)| &\le \frac{2n-1}{n^2}\Var(S) + \frac{1}{n^2} \sum_{(x,y) \in S}(y-y_s)^2
    \\ &\le \frac{2}{n} c^2 + \frac{1}{n^2} |S| c^2
\end{align}
which is at most $\frac{3}{n}c^2$, concluding the proof.
\end{proof}

\subsection{Smoothness results}
\begin{lemma}\label{lem:ED_G}
Let $G$ be a conditional $g$-gain where $g$ satisfies the hypotheses of Lemma~\ref{lem:EDR_g}. Then for all $S,S'$ with $\ED(S,S') \le 1$:
\begin{align}
    |G(S,\spl)-G(S',\spl)| \le 4 \frac{f(n)}{n}
\end{align}
where $n=\max(|S|,|S'|)$ and $f$ is as in Lemma~\ref{lem:EDR_g}.
\end{lemma}
\begin{proof}
Let $(S_0,S_1)=\spl(S)$ and $(S_0',S_1')=\spl(S')$.
Without loss of generality let $S'=S+s$ and $S_0'=S_0$, and let $n=|S'|=\max(|S|,|S'|)$. Standard calculations show that:
\begin{align}
    |G(S,\spl)-G(S',\spl)| &\le |g(S)-g(S')| + \frac{g(S_0) |S_0|}{n(n-1)} + \left|g(S_1')\frac{|S_1'|}{n} - g(S_1)\frac{|S_1|}{n-1} \right|
\end{align}
By the hypotheses of Lemma~\ref{lem:EDR_g} the first term is bounded by $\frac{f(n)}{n}$ and the second term is bounded by $\frac{f(|S_0|)|S_0|}{n(n-1)} \le \frac{f(|S_0|)}{n} \le \frac{f(n)}{n}$ too.
For the third term, since $\frac{|S_1'|}{n} > \frac{|S_1|}{n-1}$ and again by the hypotheses of Lemma~\ref{lem:EDR_g},
\begin{align}
    \left|g(S_1')\frac{|S_1'|}{n} - g(S_1)\frac{|S_1|}{n-1} \right| &\le 
    \left(g(S_1)+\frac{f(|S_1'|)}{|S_1'|}\right)\frac{|S_1'|}{n} - g(S_1)\frac{|S_1|}{n-1}
    \\ & = g(S_1)\left(\frac{|S_1'|}{n}-\frac{|S_1|}{n-1}\right) + \frac{f(|S_1'|)}{n}
    \\ &= g(S_1)\frac{n-|S_1'|}{n(n-1)} + \frac{f(|S_1'|)}{n}
    \\ &\le f(|S_1|)\frac{n-|S_1'|}{n(n-1)} + \frac{f(|S_1'|)}{n}
\end{align}
which, since $f$ is nondecreasing and $|S_1'| \ge 1$, is bounded from above by $2 \frac{f(n)}{n}$.
We conclude that $|G(S,\spl)-G(S',\spl)| \le 4 \frac{f(n)}{n}$, as claimed.
\end{proof}

\begin{lemma}\label{lem:smooth_G}
Let $G$ be a conditional $g$-gain where $g$ satisfies the hypotheses of Lemma~\ref{lem:EDR_g}. Then for all $S,S'$:
\begin{align}
    |G(S,\spl)-G(S',\spl)| \le 12 \EDR(S,S') f(\max(|S|,|S'|))
\end{align}
with $f$ is as in Lemma~\ref{lem:EDR_g}.
\end{lemma}
\begin{proof}
Fix $\spl$ and let $G(\cdot)=G(\cdot,\spl)$. By the form of $G$, we have that:
\begin{align}
    G(S) \le g(S) \le f(|S|)
\end{align}
Moreover by Lemma~\ref{lem:ED_G} we have that for all $S,S'$ with $\ED(S,S')\le 1$:
\begin{align}
    |G(S)-G(S')| \le 4 \frac{f(n)}{n}
\end{align}
where $n=\max(|S|,|S'|)$. Therefore $g_G = G$ satisfies the hypotheses of Lemma~\ref{lem:EDR_g} with $f_G = 4 f$, yielding:
\begin{align}
    |g_G(S)-g_G(S')| \le 3 \EDR(S,S') f_G(\max(|S|,|S'))
\end{align}
that is, $|G(S)-G(S')| \le 12 \EDR(S,S') f(\max(|S|,|S'))$.
\end{proof}

\begin{theorem}\label{thm:smooth_gains}
Let $\spl$ be any split rule. For all $S,S' \in (\scX \times \scY)^*$:
\begin{align}
    |\gain(S,\spl)-\gain(S',\spl)| \le \left\{
    \begin{array}{ll}
        \EDR(S,S') \cdot 48 & G=\giniGain \\
        \EDR(S,S') \cdot 60 \, \log \max(|S|,|S'|) ~~~& G=\IG \\
        \EDR(S,S') \cdot 36\, c^2 & G=\VarGain
    \end{array}
    \right.
\end{align}
where $c=\sup_{y \in \scY} |y|/2$.
\end{theorem}
\begin{proof}
By Lemma~\ref{lem:ED_gini_imp}, $G=\giniGain$ satisfies the hypotheses of Lemma~\ref{lem:smooth_G} with $g=\giniImp$ and $f=4$, hence $|\giniGain(S,\spl)-\giniGain(S',\spl)| \le \EDR(S,S') \cdot 48$.
By Lemma~\ref{lem:ED_entropy}, $G=\IG$ satisfies the hypotheses of Lemma~\ref{lem:smooth_G} with $g=H$ and $f(n) = 5 \log n$, hence $|\IG(S,\spl)-\IG(S',\spl)| \le \EDR(S,S') \cdot 60 \log \max(|S|,|S'|)$.
By Lemma~\ref{lem:ED_variance}, $G=\VarGain$ satisfies the hypotheses of Lemma~\ref{lem:smooth_G} with $g=\Var$ and $f(n) = 3c^2$, hence $|\VarGain(S,\spl)-\VarGain(S',\spl)| \le \EDR(S,S') \cdot 36 c^2$.
\end{proof}

\subsection{Approximation of maximum gain}
The next result says that, for maximum-gain decision rules, an $\epsilon$-approximate guarantees split rules whose gain is close to the maximum possible.
\begin{lemma}\label{lem:gain_approx}
Let $\splfun$ be a max-$\gain$ decision rule with $G\in\{\giniGain,\IG,\VarGain\}$. If a decision tree $T$ is $\epsilon$-approximate w.r.t.\ $(S,\splfun)$, then every internal vertex $v \in V(T)$ satisfies:
    \begin{align}
        G(S(T,v),\spl_v) \ge G(S(T,v),\spl^*_v) -
        \left\{
        \begin{array}{ll}
         \epsilon \cdot 96 & \text{if }G=\giniGain
        \\ \epsilon \cdot 120 \, \log |S(T,v)| & \text{if }G=\IG
        \\ \epsilon \cdot 72 c^2  & \text{if }G=\VarGain
        \end{array}
        \right.
    \end{align}
    where $\spl^*_v = \splfun(S(T,v))$.
\end{lemma}
\begin{proof}
Suppose $G=\giniGain$; the proof for $G=\IG$ and $G=\VarGain$ is similar. Let $T$ be $\epsilon$-approximate w.r.t.\ $(S,\splfun)$ and let $v \in V(T)$ be any internal vertex of $T$. By Definition~\ref{def:approx_tree} there exists $S_v \in (\scX \times \scY)^*$ such that $\spl_v=\splfun(S_v)$ and $\EDR(S(T,v),S_v) \le \epsilon$. Then by a double application of Theorem~\ref{thm:smooth_gains}, and since $G(S_v,\spl_v) \ge G(S_v,\spl)$ for every $\spl \in \Spl$:
\begin{align}
    G(S(T,v),\spl_v)
       &\ge G(S_v,\spl_v) - 48 \, \EDR(S(T,v),S_v)
    \\ &\ge G(S_v,\spl_v^*) - 48 \, \EDR(S(T,v),S_v) \label{eq:XXX1}
    \\ &\ge G(S(T,v),\spl_v^*) - 96 \, \EDR(S(T,v),S_v)
\end{align}
which proves the claim since $\EDR(S(T,v),S_v) \le \epsilon$.
\end{proof}

\subsection{Balancedness of threshold decision rules}
\begin{lemma}\label{lem:bal_G}
Let $G$ be a conditional $g$-gain where $g$ satisfies the hypotheses of Lemma~\ref{lem:EDR_g}. Then for every $S \in (\scX \times \scY)^*$:
\begin{align}
    \frac{\min(|S_0|,|S_1|)}{|S|} \ge \frac{G(S,\spl)}{4f(|S|)}
\end{align}
where $(S_0,S_1)=\spl(S)$ and $f$ is as in Lemma~\ref{lem:EDR_g}.
\end{lemma}
\begin{proof}
Without loss of generality assume $|S_0| \le |S_1|$. Let $\eta=\EDR(S,S_1)$ and note that $\eta=\frac{|S_0|}{|S|}$ and $1-\eta=\frac{|S_1|}{|S|}$. Then:
    \begin{align}
        \frac{|S_0|}{|S|} g(S_0) +  \frac{|S_1|}{|S|} g(S_1) & \ge (1-\eta) g(S_1)
        \\ &\ge (1-\eta)(g(S) - 3 \eta f(|S|)) && \text{Lemma~\ref{lem:EDR_g}}
        \\ &\ge g(S) - \eta(g(S)+3f(|S|))
        \\ &\ge g(S) - 4\eta f(|S|)
    \end{align}
    We conclude that $G(S,\spl) \le 4 \eta f(n)$ and therefore 
    \begin{align}
    \frac{\min(|S_0|,|S_1|)}{|S|} = \frac{|S_0|}{|S|} = \eta \ge \frac{G(S,\spl)}{4 f(n)}
    \end{align}
    concluding the proof.
\end{proof}

\begin{theorem}\label{thm:balanced_rules}
Let $\Spl$ be any family of split rules and let $G \in \{\giniGain, \IG, \VarGain\}$. Then any max-$\gain$ decision rule $\splfun : (\scX\times\scY)^* \to \Spl \cup \Lab$ with threshold $\alpha$ is $\gamma$-balanced, where:
\begin{align}
    \gamma(|S|) =
    \left\{
    \begin{array}{ll}
        \frac{\alpha}{16} & \text{if }G=\giniGain \\
        \frac{\alpha}{20 \log |S|} & \text{if }G=\IG \\
        \frac{\alpha}{12c^2} & \text{if }G=\VarGain \\
    \end{array}
    \right.
\end{align}
where $c=\sup_{y \in \scY}|y|/2$.
\end{theorem}
\begin{proof}
Let:
\begin{align}
    \begin{array}{ll}
        g=\giniImp \text{ and } f(n) = 4 & \text{if }G=\giniGain \\
        g=H \text{ and } f(n) = 5 \log n & \text{if }G=\IG \\
        g=\Var \text{ and } f(n) = 3c^2 & \text{if }G=\VarGain \\
    \end{array}
\end{align}
Then the claim follows by Lemma~\ref{lem:bal_G} by noting that, when $\spl=\splfun(S) \in \Spl$, by definition of rule with threshold $\alpha$ we have $G(S,\spl) \ge \alpha$.
\end{proof}

\section{Complexity of decision rules}\label{sec:complexity}
Recall that $\cost_{\splfun}$ is the complexity of computing $\splfun$ as a function of the length of the input. In this section we bound $\cost_{\splfun}$ for some common decision rules, proving:
\begin{theorem}\label{thm:splfun_cost}
Let $\Spl$ be the set of all split rules in the form $\spl(x)=\Ind_{x_j < t}$ or $\spl(x)=\Ind_{x_j = t}$, let $\gain \in \{\giniGain,\IG,\VarGain\}$, and let $\splfun : (\scX \times \scY) \to \Spl \cup \Lab$ be a max-$\gain$ threshold decision rule that assigns majority or average labels. Then $\cost_{\splfun}(n) \in O(d\, n \log n)$.
\end{theorem}
\noindent As a majority/average label can be computed in time $O(|S| \log |S|)$, to prove Theorem~\ref{thm:splfun_cost} it is sufficient to show that $\arg\max_{\spl \in \Spl}G(S,\spl)$ can be computed in time $O(d |S| \log |S|)$, which we do in Lemma~\ref{lem:cost_GiniGain}, Lemma~\ref{lem:cost_IG} and Lemma~\ref{lem:cost_Var}. Note that split rules in the form $\Ind_{x_j \le t}, \Ind_{x_j \ge t}, \Ind_{x_j > t}$ are captured by Theorem~\ref{thm:splfun_cost} by replacing $x_j$ with $-x_j$ and/or $\spl$ with $1-\spl$.
\newcommand{\dictN}{\text{C}}
\newcommand{\dictL}{\text{L}}
\newcommand{\dictNL}{\text{C}_{\text{L}}}
\begin{lemma}\label{lem:cost_GiniGain}
    Let $\Spl_j$ be the set of all split rules in the form $\spl(x)=\Ind_{x_j < t}$ or $\spl(x)=\Ind_{x_j = t}$. Then $\arg\max_{\spl \in \Spl_j}\giniGain(S,\spl)$ can be computed in time $O(|S| \log |S|)$.
\end{lemma}
\begin{proof}
    Suppose $\Spl_j$ is the set of all split rules in the form $\Ind_{x_j = t}$. Let $S \in (\scX \times \scY)^*$, and for every value $t$ in the domain of the $j$-th feature let $\spl_t$ be the rule defined by $\spl_t(x)=\Ind_{x_j = t}$. Let $\dictN,\dictNL,\dictL$ be associative arrays with logarithmic access/update time and linear enumeration time. 
    First, in time $O(|S| \log |S|)$, go through every $(x,y) \in S$ and increase $\dictN[x_j]$, $\dictNL[x_j][y]$, $\dictL[y]$, and compute $|S|$ and $Q=\sum_{y \in L} (\dictL[y])^2$. Then in time $O(1)$ compute:
    \begin{align}
        \giniImp(S) &= 1 - \frac{Q}{|S|^2}
    \end{align}
    Now let $t \in \dictN$. Observe that, if $(S_{t,0},S_{t,1})=\spl_t(S)$, then:
    \begin{align}
        \giniImp(S_{t,0}) = 1 - \frac{Q - \sum_{y \in \dictNL[t]} \left( (L[y])^2 - (L[y]-\dictNL[t][y])^2 \right)}{(|S|-\dictN[t])^2}
    \end{align}
    and:
    \begin{align}
        \giniImp(S_{t,1}) &= 1 - \sum_{y \in \dictNL[t]} \frac{(\dictNL[t][y])^2}{(\dictN[t])^2}
    \end{align}
    Note that $\giniImp(S_{t,0})$ and $\giniImp(S_{t,1})$ can be computed in time $O(|\dictNL[t]|\log |S|)$ by iterating on $\dictNL[t]$. Since $|\dictNL[t]| \le \dictN[t]$ and $\sum_{t} \dictN[t] = |S|$, then in time $O(|S| \log |S|)$ one can compute $\giniImp(S_{t,0})$ and $\giniImp(S_{t,1})$ for all $t \in \dictN$ and therefore (using $\giniImp(S)$ and $\dictN[t]$) also $\giniGain(S,\spl_t)$. In time $O(|S|)$ one then finds and returns $t^* = \arg\max_{t \in \dictN}\giniGain(S,\spl_t)$.

    For the case $\Ind_{x_j < t}$, sort the distinct keys of $\dictN$ by increasing value in time $O(|S|\log|S|)$; let them be $t_1 < \ldots < t_k$. For $i=1,\ldots,k$ we keep track of cumulative versions of $\dictN$ and $\dictNL$ that store:
    \begin{align}
        \dictNL^{\le}[t_i][y] &= \sum_{j=1}^i \dictNL[t_j][y]
        \\
        \dictN^{\le}[t_i] &= \sum_{j=1}^i \dictN[t_j]
    \end{align}
    Let $N_{0,0}=0$ and $N_{0,1}=1$, and for all $i=1,\ldots,k$ define:
    \begin{align}
        N_{i,0} &= \sum_{y \in L} \left(\sum_{j=1}^i \dictNL[t_j][y]\right)^2
        = \sum_{y \in L} \left(\dictNL^{\le}[t_i][y]\right)^2
        \\
        N_{i,1} &= \sum_{y \in L} \left(L[y] - \sum_{j=1}^i \dictNL[t_j][y]\right)^2 
        = \sum_{y \in L} \left(L[y]-\dictNL^{\le}[t_i][y]\right)^2
    \end{align}
    Note that, if $(S_{t_i,0},S_{t_i,1})=\spl_{t_i}(S)$, then:
    \begin{align}
        \giniImp(S_{t_i,0}) &= 1 - \frac{N_{i,0}}{\left(|S|-\dictNL^{\le}[t_i]\right)^2}
        \\
        \giniImp(S_{t_i,1}) &= 1 - \frac{N_{i,1}}{\left(\dictNL^{\le}[t_i]\right)^2}
    \end{align}
    It is not hard to compute $\dictNL^{\le}[t_{i+1}][y]$ and $\dictN^{\le}[t_{i+1}]$ from  $\dictNL^{\le}[t_{i}][y]$ and $\dictN^{\le}[t_{i}]$ in time $O(|\dictNL[t_{i+1}]|)$, and therefore to compute $N_{i,0}$ and $N_{i,1}$ and thus $\giniImp(S_{t_i,0})$ and $\giniImp(S_{t_i,1})$ for all $i$ in total time $O(|S|\log |S|)$. This implies the claim in the same way as in the previous case.
\end{proof}
\begin{lemma}\label{lem:cost_IG}
    Let $\Spl_j$ be the set of all split rules in the form $\spl(x)=\Ind_{x_j < t}$ or $\spl(x)=\Ind_{x_j = t}$. Then $\arg\max_{\spl \in \Spl_j}\IG(S,\spl)$ can be computed in time $O(|S| \log |S|)$.\end{lemma}
\begin{proof}
    Suppose $\Spl_j$ is the set of all split rules in the form $\Ind_{x_j = t}$.
    Using the same notation of Lemma~\ref{lem:cost_GiniGain}, note that:
    \begin{align}
        H(S) = \sum_{y \in \dictL} \frac{\dictL[y]}{|S|} \log \frac{|S|}{\dictL[y]}
    \end{align}
    Moreover, for every $t \in \dictN$:
    \begin{align}
        H(S_{t,1}) = \sum_{y \in \dictNL[t]} \frac{\dictNL[t][y]}{\dictN[t]} \log \frac{\dictN[t]}{\dictNL[t][y]}
    \end{align}
    Thus we can compute $H(S_{t,1})$ in time $O(|\dictNL[t]| \log |S|)$. To show that the same holds for $H(S_{t,0})$ we need some more manipulations. Note that:
    \begin{align}
        H(S_{t,0}) &= \sum_{y \in \dictL} \frac{\dictL[y] - \dictNL[t][y]}{|S| - \dictN[t]} \log \frac{|S| - \dictN[t]}{\dictL[y] - \dictNL[t][y]} \\
        &= \sum_{y \in \dictL} \frac{\dictL[y]}{|S| - \dictN[t]} \log \frac{|S| - \dictN[t]}{\dictL[y]} \\&\quad+ \sum_{y \in \dictNL[t]} \left(\frac{\dictL[y] - \dictNL[t][y]}{|S| - \dictN[t]} \log \frac{|S| - \dictN[t]}{\dictL[y] - \dictNL[t][y]} - \frac{\dictL[y]}{|S| - \dictN[t]} \log \frac{|S| - \dictN[t]}{\dictL[y]} \right) \nonumber
    \end{align}
    The second summation can clearly be computed in time $O(|\dictNL[t]| \log |S|)$. The first summation can instead be written as:
    \begin{align}
        & \frac{\log (|S| - \dictN[t])}{|S| - \dictN[t]} \sum_{y \in \dictL} \dictL[y] 
        -
        \frac{1}{|S| - \dictN[t]} \sum_{y \in \dictL} \dictL[y] \log \dictL[y]
    \end{align}
    which can be computed in time $O(\log |S|)$ if we precompute $\sum_{y \in \dictL} \dictL[y] \log \dictL[y]$ --- note that $\sum_{y \in \dictL} \dictL[y] =|S|$.  We conclude that $\IG(S,\spl_t)$ can be computed in time $O(|\dictNL[t]| \log |S|)$. Since $\sum_t |\dictNL[t]| \le |S|$, one can compute $\IG(S,\spl_t)$ for all $t \in \dictN$ in time $O(|S| \log |S|)$. In time $O(|S|)$ one then finds and returns $t^* = \arg\max_{t}\IG(S,\spl_t)$.

    The case $\Ind_{x_j < t}$ is similar, see the proof of Lemma~\ref{lem:cost_GiniGain}.
\end{proof}

\begin{lemma}\label{lem:cost_Var}
    Let $\Spl_j$ be the set of split rules in the form $\spl(x)=\Ind_{x_j < t}$ or $\spl(x)=\Ind_{x_j = t}$. Then $\spl^*_{\Spl_j}(S,\VarGain)$ can be computed in time $O(|S| \log |S|)$.
\end{lemma}
\begin{proof}
    The proof is similar to that of Lemma~\ref{lem:cost_GiniGain}.
\end{proof}

\appendix

\section{Gain measures}\label{apx:gains}
We recall the definitions of Gini gain, information gain and variance gain. 

\subsection{Gini gain}
Let $\scY=\{1,\ldots,k\}$. The \emph{Gini Impurity} of $S\in(\scX\times\scY)^*$ is:
\begin{align}
    \giniImp(S) = 1 - \sum_{i=1}^k p_i^2(S)
\end{align}
where $p_i(S)$ is the fraction of examples of $S$ having label $i$.
The \emph{Gini gain} $\giniGain : (\scX \times \scY)^* \times \Spl$ is the conditional $\giniImp$-gain.

\subsection{Information gain}
Let $X \in \N$ be a random variable. The entropy of $X$ is:
\begin{align}
    H(X) = \sum_{x} \Pr(X=x) \log \frac{1}{\Pr(X=x)}
\end{align}
If $X$ takes on at most $k$ distinct values with positive probability then $H(X) \le \log k$. For $p \in [0,1]$ let $H(p)$ be the entropy of a Bernoulli random variable of parameter $p$. Let $Y,Z \in \N$ be two random variables defined on the same space of events. The conditional entropy of $Y$ given $Z$ is:
\begin{align}
    H(Y|Z) = \sum_{z} \Pr(Z=z) H(Y|Z=z)
\end{align}
The \emph{mutual information} or \emph{information gain} between $Y$ and $Z$ is:
\begin{align}
    \IG(Y,Z) = H(Y) - H(Y|Z) = H(Z) - H(Z|Y) = H(Y)+H(Z)-H((Y,Z)) 
\end{align}

Now let $S \in (\scX \times \scY)^*$. The entropy of $S$ and the information gain of a split rule $\spl$ on $S$ are defined as follows: letting $(X,Y)$ be a random uniform element of $S$,
\begin{align}
    H(S) &=H(Y)   
    \\
    \IG(S,\spl) &=\IG(Y,\spl(X))
\end{align}
One can see that $\IG : (\scX \times \scY)^* \times \Spl$ is the conditional $H$-gain.

\subsection{Variance gain}
Let $\scY = \R$. The variance of the labels of a set $S$ is:
\begin{align}
    \Var(S) = \frac{1}{|S|^2}\sum_{s,s'\in S} (y-y')^2
\end{align}
where $s=(x,y)$ and $s'=(x',y')$. The \emph{variance gain} $\VarGain : (\scX \times \scY)^* \times \Spl$ is the conditional $\Var$-gain.

\section{Proofs for Section~\ref{sec:dyn}}
\label{apx:approx}
\subsection{Proof of Theorem~\ref{thm:fudy_eps_feasible}}
We prove a slightly different statement, from which Theorem~\ref{thm:fudy_eps_feasible} follows by substituting $\epsilon$.
\begin{theorem}\label{thm:fudy_eps_feasible_full}
    Let $\Spl$ be the set of split rules in the form $\spl(x)=\Ind_{x_j < t}$ or $\spl(x)=\Ind_{x_j = t}$, let $G\in\{\giniGain,\IG,\VarGain\}$, and let $\beps=(\alpha,\beta) \in (0,1]^2$.
    One can maintain an $\beps$-feasible decision tree using $O\!\left( \frac{d \log^4 n}{\epsilon^3} \right)$ operations per update request in the worst case, where
    \begin{align*}
        \epsilon = \left\{
        \begin{array}{ll}
            \frac{\min(\alpha,\beta)}{100} & G=\giniGain \\[5pt]
            \frac{\min(\alpha,\beta)}{130 \log n} ~~& G=\IG \\[5pt]
            \frac{\min(\alpha,\beta)}{80c^2} & G=\VarGain, \; c = \sup_{y \in \scY} |y|/2
        \end{array}
        \right.
    \end{align*}
\end{theorem}
\begin{proof}
    Let us first consider the case $G=\giniGain$. Define $\splfun$ to be a max-gain decision rule with threshold $\frac{\alpha}{2}$ that assigns majority/average labels. By Theorem~\ref{thm:balanced_rules}, $\splfun$ is $\gamma$-balanced for $\gamma=\frac{\alpha}{32}$. By Theorem~\ref{thm:splfun_cost}, $\cost_{\splfun}(n) = O(d n \log n)$. Let $\epsilon=\frac{\min(\alpha,\beta)}{100}$; clearly $\epsilon < \frac{\gamma}{2}$ and $\gamma-2\epsilon=\Theta(\epsilon)$. By Theorem~\ref{thm:main}, there is an algorithm that maintains an $\epsilon$-approximate tree using $O\!\left(\frac{d \log^4 n}{\epsilon^3}\right)$ operations per update request. It is straightforward to keep at every leaf $v$ of $T$ a priority queue that maps every label to its count and, in constant time, returns a majority label or average label of $S(T,v)$ without altering the running time.
    To show that the tree maintained by the algorithm is $\beps$-feasible, consider the current active set $S^i$ and the current tree $T^i$ and apply Theorem~\ref{thm:eps_to_beps}.

    The proof for $G=\IG$ is similar. Define $\splfun$ as above. By Theorem~\ref{thm:balanced_rules}, $\splfun$ is $\gamma$-balanced where $\gamma(n)=\frac{\alpha}{40 \log n}$, by Theorem~\ref{thm:splfun_cost} $\cost_{\splfun}(n) = O(d n \log n)$, and letting $\epsilon=\frac{\min(\alpha,\beta)}{130 \log n}$ ensures $\epsilon < \frac{\gamma}{2}$ and $\gamma-2\epsilon=\Theta(\epsilon)$. The $O\!\left(\frac{d \log^4 n}{\epsilon^3}\right)$ bound and the $\beps$-feasibility follows like above.
    
    The proof for $G=\VarGain$ is completely analogous.
\end{proof}

\section{Ancillary results}
\begin{lemma}\label{lem:greedy}
$\greedy_{\splfun}(S)$ runs in time $O\!\left(\hat h \cdot \Big(\cost_{\splfun}(|S|) + d\,|S|\log|S| \Big)\right)$ where $\hat h$ is the height of the returned tree.
\end{lemma}
\begin{proof}
For each $v\in V(T)$, $\greedy_{\splfun}$ spends time $\cost_{\splfun}(S(T,v)) + O(d|S(T,v)|\log |S(T,v)|)$ to compute $\splfun(S(T,v))$ and populate $D(T,v)$. If $v$ is internal then $\greedy_{\splfun}$ also spends time $O(d|S(T,v)|)$ to compute $(S(T,v_1),S(T,v_2))=\spl(S(T,v))$ where $v_1,v_2$ are the children of $v$; this is dominated by the bound above. The claim follows by the assumptions on $\cost_{\splfun}$ and since $\sum_{v \in V_i} |S(T,v)| \le |S|$ where $V_i$ by the subset of $V$ at depth $i$.
\end{proof}

\begin{lemma}\label{lem:balanced_approx}
If $T$ is $\epsilon$-approximate w.r.t.\ $(S,\splfun)$ and $\splfun$ is $\gamma$-balanced, then every $v,w \in V(T)$ with $w$ child of $v$ satisfy $|S(T,w)| \ge (\gamma - 2\epsilon) |S(T,v)|$, and therefore $h(T) = O\!\left(\frac{\log |S|}{\gamma-2\epsilon}\right)$.
\end{lemma}
\begin{proof}
Let $S_v$ be the set for which $v$ satisfies the definition of $\epsilon$-approximate tree. Then:
\begin{align}
|S(T,w)| &\ge |S_v(T,w)| - \ED(S_v,S(T,v))
\\ &\ge |S_v(T,w)| - \epsilon |S(T,v)|
\\ &\ge \gamma |S_v| - \epsilon |S(T,v)|
\\ &\ge \gamma (1-\epsilon) |S(T,v)| - \epsilon |S(T,v)|
\\ &\ge (\gamma - 2\epsilon) |S(T,v)|
\end{align}
This implies $|S(T,w)| \le (1-(\gamma-2\epsilon))|S(T,v)|$ and thus the claim on $h(T)$.
\end{proof}

\bibliographystyle{plain}
\bibliography{references}

\end{document}